\PassOptionsToPackage{table}{xcolor}
\documentclass[acmsmall]{acmart}

\usepackage[table]{xcolor}
\usepackage{comment}
\usepackage[linesnumbered,ruled,vlined]{algorithm2e}%[ruled,vlined]{  
\usepackage{algpseudocode}  
\usepackage{textgreek}
\usepackage{tablefootnote}
\usepackage{multirow}
\usepackage{overpic}
\usepackage{subfigure}
\usepackage{enumitem}
\usepackage{bm}
\usepackage{tabu}
\usepackage{graphicx}
\usepackage{booktabs}
\usepackage{array}
\usepackage{etoolbox}
\usepackage{subfloat}

\newcommand\ourmethod{\textsc{Vertiorizon}}

\definecolor{ao(english)}{rgb}{0.0, 0.5, 0.0}
\newcounter{dingheng}
\numberwithin{dingheng}{section}

\newcounter{siqiang}
\numberwithin{siqiang}{section}

\AtBeginDocument{%
  \providecommand\BibTeX{{%
    \normalfont B\kern-0.5em{\scshape i\kern-0.25em b}\kern-0.8em\TeX}}}

\newcommand\horizontal{horizontal scheme}
\newcommand\leveling{vertical scheme}
\newcommand\Horizontal{Horizontal Scheme}
\newcommand\Leveling{Vertical Scheme}

\newcommand\scheme{growth scheme}
\newcommand\Scheme{Growth Scheme}

\newtheorem{problem}{Problem}

\setcopyright{rightsretained}
\acmJournal{PACMMOD}
\acmYear{2025} \acmVolume{3} \acmNumber{3 (SIGMOD)} \acmArticle{173} \acmMonth{6} \acmPrice{}\acmDOI{10.1145/XXXXXXX}

\makeatletter
\gdef\@copyrightpermission{
  \begin{minipage}{0.2\columnwidth}
   \href{https://creativecommons.org/licenses/by/4.0/}{\includegraphics[width=0.90\textwidth]{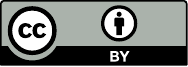}}
  \end{minipage}\hfill
  \begin{minipage}{0.8\columnwidth}
   \href{https://creativecommons.org/licenses/by/4.0/}{This work is licensed under a Creative Commons Attribution International 4.0 License.}
  \end{minipage}
  \vspace{5pt}
}
\makeatother


\begin{document}

%\title{LLL: Light-weight Learning-based LSM-tree Tuning}
%\title{Enhancing LSM-tree with Binomial Merging}
%\title{How to Grow Your LSM-trees?}
\title{How to Grow an LSM-tree? }
\subtitle{Towards Bridging the Gap Between Theory and Practice}

% \affiliation{%

\author{Dingheng Mo}
\affiliation{%
  \institution{Nanyang Technological University}
  \country{Singapore}
}
\email{dingheng001@ntu.edu.sg}

\author{Siqiang Luo}
%\authornote{Corresponding Author}
\affiliation{%
  \institution{Nanyang Technological University}
  \country{Singapore}
}
\email{siqiang.luo@ntu.edu.sg}

\author{Stratos Idreos}
\affiliation{%
  \institution{Harvard University}
  \country{USA}
}
\email{stratos@seas.harvard.edu}

\begin{abstract}
LSM-tree based key-value stores are widely adopted as the data storage backend in modern big data applications. 
The LSM-tree {\it grows} with data ingestion, by either adding levels with fixed level capacities (dubbed as vertical scheme) or increasing level capacities with fixed number of levels (dubbed as horizontal scheme). 
The vertical scheme leads the trend in recent system designs in RocksDB, LevelDB, and WiredTiger, whereas the horizontal scheme shows a decline in being adopted in the industry. The growth scheme profoundly impacts the LSM system performance in various aspects such as read, write and space costs. This paper attempts to give a new insight into a fundamental design question -- {{\it how to grow an LSM-tree} to attain more desirable performance?}

Our analysis highlights the limitations of the vertical scheme in achieving an optimal read-write trade-off and the horizontal scheme in managing space cost effectively. Building on the analysis, we present a novel approach, {\ourmethod}, which combines the strengths of both the vertical and horizontal schemes to achieve a superior balance between lookup, update, and space costs. Its adaptive design makes it highly compatible with a wide spectrum of workloads. Compared to the vertical scheme, {\ourmethod} significantly improves the read-write performance trade-off. In contrast to the horizontal scheme, {\ourmethod} greatly extends the trade-off range by a non-trivial generalization of Bentley and Saxe's theory~\cite{bentley1980decomposable}, while substantially reducing space costs. When integrated with RocksDB, {\ourmethod} demonstrates better write performance than the vertical scheme, while incurring about six times less additional space cost compared to the horizontal scheme.

\end{abstract}

\maketitle

\section{Introduction}
\label{sec: intro}
%\noindent\textbf{LSM-Based Storage Systems are Prevalent.}
%LSM-tree appends new data entries into the main memory buffer, whose filled-up will sort the buffered entries and compact them as a larger sorted run to the secondary storage. 
LSM-tree is a key-value data structure that organizes data into multiple increasing levels. %with the initial level being a memory buffer and the other levels on the disk.
%LSM-tree is a key-value data structure that organizes data into multiple increasing tiers, starting with a buffer in the memory, followed by a series of levels on the disk.
It appends incoming data entries into the buffer, whose filled-up triggers a merge-sort of the buffer into the next level. %, forming a larger sorted sequence that resides in the secondary storage.
The process can cascade down whenever the affected level reaches its capacity to ensure desired read performance and minimize the space amplification.
% Reserch on optimizing LSM-trees has gained significant popularity in recent years. 
% \vspace{1mm}
% \noindent\textbf{growth scheme\textemdash the Fundamental of LSM-Tree.} \dingheng{Revise this paragraph.}
%The LSM-tree presents an alluring solution for data-intensive applications by converting random disk writes into sequential writes. %This foundational principle is realized by organizing data on the disk into multiple levels within a hierarchy of incrementally increasing size. 
Due to its desired performance, the LSM-tree is extensively employed as the backbone of various data storage systems. This includes its utilization in key-value stores (KV-stores) like RocksDB~\cite{rocksdb},  Cassandra~\cite{lakshman2010cassandra}, and ScyllaDB~\cite{scylladb}, NewSQL stores such as CockroachDB~\cite{cockroachdb}, as well as time-series databases like InfluxDB~\cite{influxdb}.

\begin{figure}[t]
    \centering
    \includegraphics[width=0.6\textwidth]{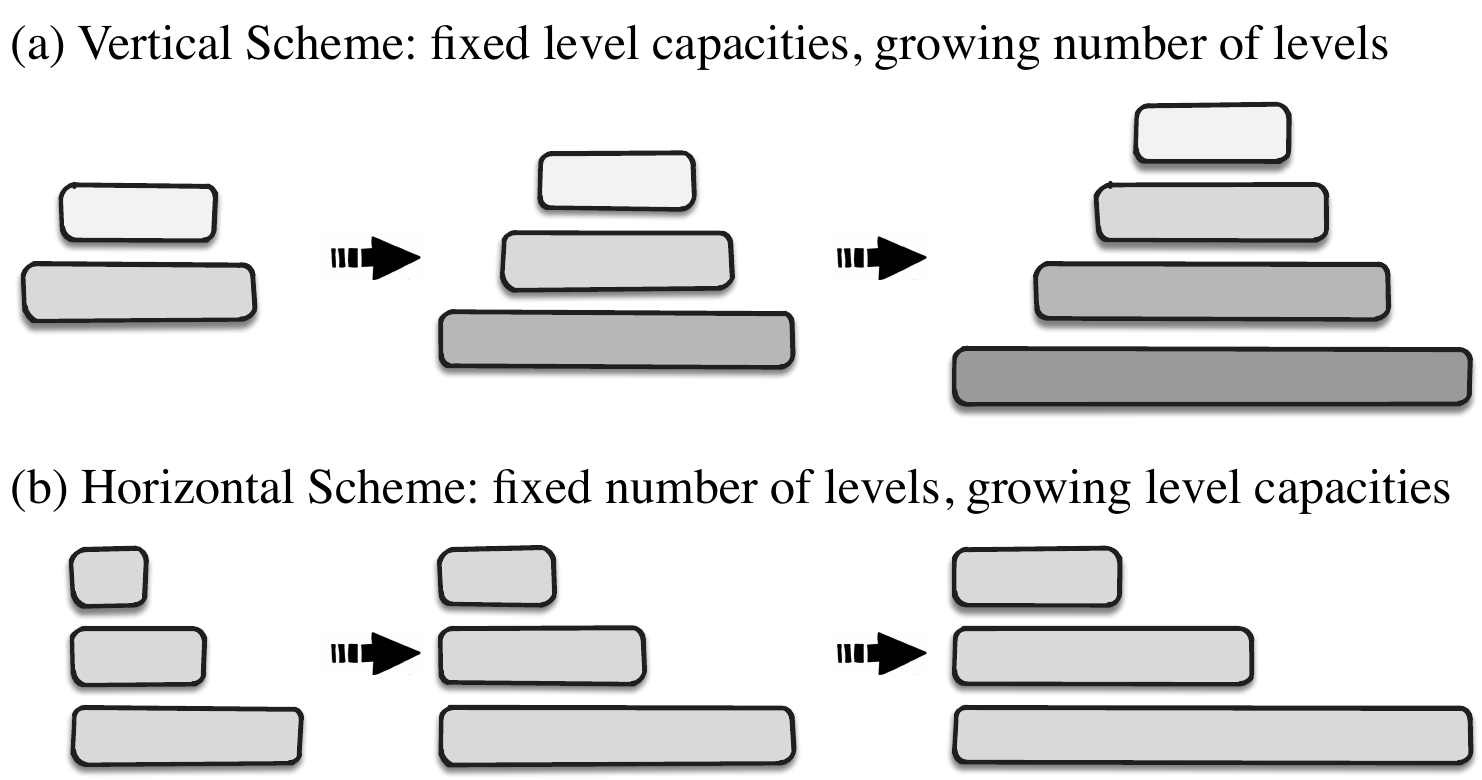}
    \vspace{0mm}
    \caption{Different growth schemes of LSM-trees.}
    \label{fig: intuition}
    \vspace{-2mm}
\end{figure}

% \begin{table}[b]
% %\vspace{-2mm}
% \renewcommand\arraystretch{1.2}
% \scalebox{0.83}{
% \begin{tabular}{|c|c|c|}
% \hline
% \textbf{Growth Scheme}                                                      & Vertical                                                                           & Horizontal                                                           \\ \midrule
% \textbf{\begin{tabular}[c]{@{}c@{}}Adopted KV \\ Systems\end{tabular}} & \begin{tabular}[c]{@{}c@{}} LevelDB, RocksDB, WiredTiger, \\ Cassandra, ScyllaDB, BadgerDB ...\end{tabular} & \begin{tabular}[c]{@{}c@{}}BigTable, HBase,\\ AsterixDB\end{tabular} \\ \hline
% %\textbf{Constrain}                                                           & Level capacity = $T\cdot B^i$                                                      & Number of levels = $\ell$                                               \\ \hline
% % \textbf{\begin{tabular}[c]{@{}c@{}}Compaction \\ Granularity\end{tabular}}   & \begin{tabular}[c]{@{}c@{}}Partial and Full\\ (Mostly Partial)\end{tabular}        & Full only                                                            \\ \hline
% \end{tabular}
% }
% \caption{The landscape of the two growth schemes.}
% %\vspace{-2mm}
% \label{table: summary}
% %\vspace{2mm}
% \end{table}

\vspace{1mm}
\noindent\textbf{Vertical and Horizontal Growth Schemes.}
The growth scheme governs how LSM-trees expand and maintain their hierarchical structure amidst the continuous influx of new data entries.
Specifically, the LSM-tree relies on compactions to transfer data from one level to the next, ensuring that each level’s size is significantly larger than the preceding one. The growth scheme fundamentally controls the timing of these compactions at each level.
Currently, the growth schemes adopted in mainstream LSM-based key-value stores fall into two major categories, the {\leveling} and the {\horizontal}.
With the {\leveling}, data ingestion into the LSM-tree would expand it vertically by creating additional levels. Specifically, each LSM-tree level has a fixed capacity that increases exponentially with the number of levels, as Figure~\ref{fig: intuition} (a) depicts. Compactions are initiated from one level to the subsequent level as soon as the level is full.
Under the {\horizontal}, the total number of levels on the disk is confined to a fixed value. As illustrated in Figure~\ref{fig: intuition} (b), each level grows horizontally with a heuristic (see Section~\ref{section: analysis}) that generally maintains the relative capacity between levels, implying that their capacities would gradually expand as data flows in. Beyond the last level, each level initiates a compaction into the subsequent level after a certain number of compactions from its preceding level. The required number of compactions gradually increases as the levels grow. 
% Compactions under {\horizontal} merge multiple levels together to reduce the total number of levels, unlike in the {\leveling} where data normally shifts from one level to the next level.
%\siqiang{We can make a summary table.}

\begin{table}[h]
\vspace{-1mm}
\renewcommand\arraystretch{1.2}
\setlength{\abovecaptionskip}{-0.5cm}
\setlength{\belowcaptionskip}{0cm}
\scalebox{0.92}{
\begin{tabular}{|c|c|c|}
\hline
\textbf{Growth Scheme} & \textbf{Vertical} & \textbf{Horizontal} \\ \midrule
\textbf{Leveling Merge} & \begin{tabular}[c]{@{}c@{}}LevelDB, RocksDB, \\ WiredTiger, BadgerDB ...\end{tabular} & \begin{tabular}[c]{@{}c@{}}BigTable, HBase,\\ AsterixDB\end{tabular} \\ \hline
\textbf{Tiering Merge} & Cassandra, ScyllaDB & Not exist \\ \hline
\end{tabular}
}
\vspace{0.7mm}
\caption{The landscape of the two growth schemes.}\label{tab:vacancy}
\vspace{0.8mm}
\end{table}

\vspace{1mm}
\noindent\textbf{{The Decline of Horizontal Scheme V.S. Its Desired Theoretical Property.}} 
The horizontal scheme had been adopted by many pioneering key-value storage systems, including Bigtable ~\cite{chang2008bigtable}, HBase~\cite{george2011hbase}, and AsterixDB~\cite{alsubaiee2014asterixdb}. 
However, a noteworthy trend in recent ten years is that most of the new generation key-value storage systems employ the vertical scheme, such as LevelDB~\cite{leveldb}, RocksDB~\cite{rocksdb}, WiredTiger~\cite{wiredtiger}, Cassandra~\cite{lakshman2010cassandra}, ScyllaDB~\cite{scylladb}, and BadgerDB~\cite{badgerdb}, as Table~\ref{tab:vacancy} summarizes.
Furthermore, a majority of the prominent NoSQL databases in recent times, such as CockroachDB~\cite{cockroachdb}, TiDB~\cite{tidb}, Dgraph~\cite{dgraph}, and PolarDB~\cite{polardb}, have opted for the vertical scheme in their storage backend. This is probably attributed to that the vertical scheme offers the advantage of being easier to manage, and it allows for a finer compaction granularity to mitigate space amplification and write stalls. Contrary to this trend, interestingly, an early theory back in 1980's~\cite{bentley1980decomposable} revealed that the horizontal-growth-style data structure could offer an optimal trade-off between read and write costs even before the formal invention of the LSM-tree, and this theory is further explicitly discussed in the context of LSM-trees by Mathieu et al.~\cite{mathieu2014bigtable}. This controversy strongly indicates that the horizontal growth scheme might be undervalued in practical system design choices.

\vspace{1mm}
\noindent
{\color{black}
{\bf The Problem: The Imperfect LSM-tree Growth Scheme Hurts System Performance.} How to grow (a.k.a., expand) the LSM-trees with fast data writes is among the most critical design choices, profoundly impacting the performance of LSM-tree based systems. {For instance, in a cloud environment using LSM-based storage, the chosen growth scheme impacts the system latency, write amplification, and space cost during data ingestion. A suboptimal growth scheme leads to higher write amplification, requiring more disk scans, which shortens hardware lifespan and raises cloud costs. Additionally, ineffective growth schemes result in increased latency and space usage, driving up costs for users storing and querying data. 
This is because the growth scheme dictates the timing of compactions, and inappropriate compaction timing can lead to higher write and read amplification (as we detail later in Sections 3 and 4), thereby increasing latency. Furthermore, different schemes support varying compaction granularities, while full compaction in particular may cause high space amplification, resulting in increased storage cost. Existing growth schemes thus face limitations that hinder system scalability in various ways. 
The vertical scheme, used by RocksDB and LevelDB, is not optimal in balancing read-write trade-offs because it performs compaction at each level with a fixed frequency, leading to reduced system throughput as data scales.
%Unfortunately, the existing growth schemes have their limitations, hindering system scalability in various ways. 
The horizontal scheme, employed by BigTable, has only been implemented with the leveling merge policy, making it difficult to handle update-intensive workloads—a common workload type in modern systems. It also requires full compaction during merges, resulting in high space overhead. This motivates us to revisit the strengths and weaknesses of the two growth schemes, and ask for the possibility of better designs.}}

{\small
\begin{table}[tb]
\vspace{1mm}
\setlength{\abovecaptionskip}{-0.5cm}
\setlength{\belowcaptionskip}{0cm}
\renewcommand{\arraystretch}{1.05}
\setlength{\tabcolsep}{3.5pt}
\scalebox{0.9}{
\begin{tabular}{|c|c|c|c|c|c|}
\hline
\multicolumn{1}{|l|}{\textbf{Designs}} &
  \begin{tabular}[c]{@{}c@{}}{Vertical-}\\{Leveling}\end{tabular} &
  \begin{tabular}[c]{@{}c@{}}{Vertical-}\\{Tiering}\end{tabular} &
  \begin{tabular}[c]{@{}c@{}}{Horizontal-}\\{Leveling}\end{tabular} &
  \begin{tabular}[c]{@{}c@{}}{Horizontal-}\\{Tiering {\bf{(ours)}}}\end{tabular} &  
  \begin{tabular}[c]{@{}c@{}}{\ourmethod}\\{\bf{(ours)}}\end{tabular} \\ \hline
 \begin{tabular}[c]{@{}c@{}}{Update-}\\{Heavy}\end{tabular} &
  \scalebox{0.8}{$\bigstar$} &
  \scalebox{0.8}{$\bigstar\bigstar\bigstar\bigstar$} &
  \scalebox{0.8}{$\bigstar\bigstar\bigstar$} &
  \scalebox{0.8}{$\bigstar\bigstar\bigstar\bigstar\bigstar$} &
  \scalebox{0.8}{$\bigstar\bigstar\bigstar\bigstar\bigstar$} \\ \hline
 \begin{tabular}[c]{@{}c@{}}{Lookup-}\\{heavy}\end{tabular} &
  \scalebox{0.8}{$\bigstar\bigstar\bigstar\bigstar\bigstar$} &
  \scalebox{0.8}{$\bigstar\bigstar$} &
  \scalebox{0.8}{$\bigstar\bigstar\bigstar\bigstar\bigstar$} &
  \scalebox{0.8}{$\bigstar\bigstar\bigstar$} &
  \scalebox{0.8}{$\bigstar\bigstar\bigstar\bigstar\bigstar$} \\ \hline
Space &
  \scalebox{0.8}{$\bigstar\bigstar\bigstar\bigstar$} &
  \scalebox{0.8}{$\bigstar\bigstar\bigstar\bigstar$} &
  \scalebox{0.8}{$\bigstar\bigstar\bigstar$}&
  \scalebox{0.8}{$\bigstar\bigstar\bigstar$} &
  \scalebox{0.8}{$\bigstar\bigstar\bigstar\bigstar$} \\ \hline
\end{tabular}
%\vspace{-2mm}
}
\vspace{1mm}
\captionof{table}{Comparing existing growth schemes and ours.}%Horizontal-tiering has the best performance under update-heavy scenarios by achieving optimal read-write trade-off under the tiering merge policy. Furthermore, {\ourmethod} always provides strong performance by self-navigating the horizontal scheme. }
%\vspace{-3mm}
\label{tab:structure_analysis}
\vspace{-1mm}
\end{table}
}

\vspace{1mm}
\noindent
{\bf Key Intuition 1.}
{There is an increasing need for a systematic discussion on the foundational topic of LSM-tree growth schemes, due to the growing scale and complexity of LSM-based systems.}
% It is time to revoke a systematic discussion on this foundational topic -- LSM-tree growth scheme -- also partly owes to a more complicated picture of LSM-based systems today. 
One important recent design is to adopt the tiering merge policy (i.e., each level is compacted as a new sorted run in the subsequent level), rather than the default leveling merge policy (i.e., each level is compacted to the existing sorted run in the subsequent level), to improve write performance at the expense of read performance. %Recall that the leveling policy maintains a single sorted run at each level, while the tiering policy permits multiple sorted runs within a level. 
As demonstrated by Dayan and Idreos~\cite{dostoevsky2018}, the vertical growth scheme can be easily combined with the tiering policy, increasing its flexibility in balancing read-write trade-offs. However, the integration of the horizontal scheme with the tiering merge policy remains unexplored, and no such design has been adopted previously (Table~\ref{tab:vacancy}). As a result, the horizontal scheme offers limited flexibility in balancing read and write performance, which is particularly undesirable for workloads with frequent updates. Our first intuition is, {\it what if we can design a horizontal-tiering scheme that couples horizontal scheme with tiering merge policy?} Such a design, if feasible, could significantly elevate the status of the horizontal scheme and revive interest in this largely overlooked approach. However, creating an effective horizontal-tiering scheme is notoriously difficult and has not been accomplished since the invention of LSM-trees. %To explain, the horizontal scheme has a careful setup of compaction timings in order to achieve a desired read-write trade-off, and this properly never holds if we do the same as that in~\cite{dostoevsky2018} to turn it into tiering merges. 

\vspace{1mm}
\noindent
{\bf Key Intuition 2. } Since the strengths and weaknesses of the two major growth schemes complement each other, our second intuition is {\it whether it is possible to combine the two to achieve the best of both worlds}. There are countless potential ways to hybridize these schemes, but ensuring the combination fully leverages the strengths of both is challenging. Moreover, the optimal growth scheme often depends on the specific workload. For example, a read-intensive workload might favor a leveling-based scheme over a tiering-based one. This implies the necessity of a growth scheme that can be adaptive to different workloads, adding to the sophistication in hybridizing the schemes.

\vspace{1mm}
\noindent\textbf{{Our Proposal.}} 
%Motivated by the preceding questions and challenges, we introduce {\Ourmethod}: \textbf{H}ybrid G\textbf{r}ow\textbf{i}ng Scheme for Key-\textbf{V}alue Stor\textbf{e}s. {\ourmethod} is a novel and efficient LSM-tree growth scheme that melds the strength of both vertical and horizontal schemes. Specifically, 
We present two novel techniques, namely Horizontal-Tiering and {\ourmethod}
%~\footnote{A word hybrids {\it vertical} and {\it horizontal}.}
, to address the challenges. Horizontal-tiering enhances the horizontal scheme with the integration of the tiering policy (response to Intuition 1). This technique will revoke the recognition of the largely overlooked horizontal scheme, and thus highly impact the landscape of LSM-tree growth schemes. On top of this enhanced scheme, we further propose a new design {\ourmethod}, which is a novel self-designing LSM-tree growth scheme that harnesses the strengths of both vertical and horizontal schemes, obtaining both theoretical merits and desired practical performance (response to Intuition 2). As shown in Figure~\ref{exp: tradeoff and embed}~(a), Horizontal-Tiering effectively expands the optimal read-write trade-off to a much wider spectrum. As shown in Table~\ref{tab:structure_analysis}, compared to existing schemes like Vertical-Leveling, Vertical-Tiering, and Horizontal-Leveling, %Horizontal-Tiering offers the best update cost while maintaining moderate lookup and space costs. 
{\ourmethod} offers very robust performance across all dimensions. We summarize our contributions as follows. %This newly designed scheme effectively expands the horizontal scheme's range of optimal read-write performance trade-offs, as Figure~\ref{exp: tradeoff and embed}~(a) shows. %It can be easily integrated into a full KV-system such as RocksDB. 
\begin{itemize}[leftmargin=*]
\vspace{-1mm}
    \item  
    We carefully review the pros and cons of the two major growth schemes for LSM-trees, and motivates the fundamental problem in designing better growth schemes. %contend that the vertical and horizontal schemes are the predominant growth schemes in modern LSM-based key-value databases. Furthermore, we suggest that the horizontal scheme should intuitively offer a superior read-write trade-off compared to the vertical scheme.
%    Through in-depth theoretical analysis and hands-on experiments, we demonstrate that {\leveling} suffers from sub-optimal write performance while {\horizontal} encounters excessive space amplification.
    \item We present the horizontal-tiering scheme, which is the first horizontal scheme that can be effectively applied under the tiering merge policy since the invention of LSM-trees. %inprovides a superior and more flexible read-write trade-off under tiering merge policy, effectively boosting LSM-tree performance under write-heavy workloads. 
    We also theoretically show that our designed horizontal-tiering scheme achieves the minimum accumulative cost during compaction. This is a non-trivial generalization of Bentley and Saxe's theory~\cite{bentley1980decomposable}.%We theoretically analyze the read and write trade-off and space amplification in the vertical and horizontal schemes. 

    \item   We present {\ourmethod}, an innovative growth scheme that embodies the strengths of both the vertical and horizontal schemes, which also offers a nice self-designing feature to work in an adaptive manner with different workloads.
    %\vspace{1mm}

    \item We design a novel algorithm to adapt our design to work under more skewed workloads. This is accompanied with a deep analysis to justify the rationale behind such designs.

    \item  We present typically how {\ourmethod} can be embedded in a mature LSM-system by introducing a novel dynamic Bloom filter layout tailored to our design. 
    %Our dynamic Bloom filter layout achieves improved performance under the hybrid scheme compared with state-of-the-art Bloom filter layouts.
    %\vspace{1mm}

 %   \item We design a cushion level structure to address a specific yet common issue that leads to performance degradation when compaction occurs between two levels having different compaction granularities.
   % \vspace{1mm}

    \item We carry out extensive experimental assessments of {\ourmethod} in real-world system environments. We demonstrate that {\ourmethod} achieves up to 3.2 times throughput than the vertical scheme (employed by RocksDB), and incurs about 6 times less additional space amplification compared to the horizontal scheme.
   % \vspace{1mm}

\end{itemize}

\vspace{-3mm}
\section{Background}
\label{section: background}

\noindent\textbf{LSM-tree.}
The LSM-tree~\cite{o1996log} is a persistent multi-level index structure designed to store key-value entries. 
An LSM-tree achieves efficient write performance by transforming random writes to the disk into sequential writes.
Specifically, when a key-value entry is inserted, the LSM-tree first temporarily stores it in a main memory buffer. Once the buffer reaches its capacity, all entries within it are merged and sorted by keys into the first level on the disk. 
%A level in the LSM-tree stores and indexes a sequence of completely sorted entries on disk.
%If there was only one level, every time new data entries were merged and sorted from the memory buffer into this level, it would cause significant write amplification. This is because all data from this level would need to be retrieved from the disk and then written back.
To gain better write performance,
LSM-tree organizes data across various sorted levels on the disk, arranged in ascending order of size. When certain conditions are met, the LSM-tree would merge one level into the subsequent level through compaction, continuing in this manner up to the largest level. In this way, the write amplification associated with merging data from the buffer or a specific level to the subsequent one is duly mitigated.
Write amplification quantifies the average count of physical rewrites for each entry, due to the recurrent data rewrites to the storage driven by the compaction process inherent to the LSM-tree.
Read amplification measures the average number of disk pages read for a single lookup, as a search might scan through several sorted runs within an LSM-tree to find a specific key. Space amplification evaluates the additional space cost needed on average for each unique entry.
%This compaction process involves the merge of data in Level $i$ into a larger sorted run in Level $i+1$. 
%Generally, compactions enhance the read efficiency of an LSM-tree by reducing the number of sorted runs to probe during each query and decreasing space through the elimination of redundant entries. Yet, the compaction process itself is resource-intensive, influencing write performance adversely.
%Compaction improves the read efficiency in LSM-tree because it reduces the sorted runs to be probed in each query and cuts down space by eliminating redundant entries. However, compaction itself is a costly process.
%Consequently, the capacity of Level $i+1$ is $T$ times greater than that of Level $i$, where $T$ represents the capacity ratio.

% \vspace{1mm}
% \noindent\textbf{Amplification.} 
% %Read amplification, write amplification, and space amplification are key metrics for assessing the efficiency of an LSM-tree. 
% Read amplification measures the average number of disk pages read for a single lookup, as a search might scan through several sorted runs within an LSM-tree to find a specific key. Write amplification quantifies the average count of physical rewrites for each entry, due to the recurrent data rewrites to the storage driven by the compaction process inherent to the LSM-tree. Space amplification evaluates the additional space cost needed on average for each unique entry.
%\dingheng{Add space amplification.}

\vspace{0.5mm}
\noindent\textbf{Full and Partial Compactions.}
%Compaction granularity refers to the amount of data that gets moved to the next level when a compaction occurs. 
Usually, a compaction merges and sorts the full data from one level with the next, known as full compaction. However, in LSM-trees, a deeper level could hold vast amounts of data, and full compaction in such cases can cause major performance issues. %because it requires loading and processing massive data volumes from two adjacent levels. 
This not only slows down the system but also leads to substantial space amplification, as original data cannot be disposed of until compaction is fully done. The partial compaction mitigates these problems, allowing merging only a small fraction of the data with the subsequent level in each compaction. % which helps avoid large performance drops and excessive use of space. However, this benefit comes at the cost of slightly higher write amplification.

\vspace{0.5mm}
\noindent\textbf{Tiering and Leveling Merge Policy.}
%Leveling and tiering are the two representative merge policies of LSM-trees. 
Under the leveling merge policy, a compaction from Level \( i \) to Level \( i+1 \) merge-sorts all data from Level \( i \) with all the existing data in Level \( i+1 \). Additionally, each buffer flush merge-sorts its data with the existing data in the first level. In contrast, the tiering merge policy handles compactions or buffer flushes to a level by adding incoming data as a new separate sorted run in that level, without merging it with the existing data there. In this way, each entry is involved in only one compaction at each level. Generally, the leveling merge policy is more read-friendly, while tiering performs better in write-intensive scenarios. %An LSM-tree with the vertical scheme allows switching between these two merge policies to trade between lookup cost and update cost. However, the existing design of the horizontal scheme only accommodates leveling merge and does not support tiering.

\vspace{0.5mm}
\noindent\textbf{Point Lookup and Range Lookup.}
Point-lookup locates the entry for a given key. During the lookup, the sorted run is probed from the newest to the oldest, until the entry containing the queried key is found. Modern LSM-based systems typically store the starting key of each disk block as an index in memory. Since each LSM-tree run is sorted, this setup allows for a direct pinpoint of the sole block likely to contain a particular key, ensuring a maximum of one I/O for each lookup at every sorted run. This process is further accelerated by applying Bloom filter at each sorted run. The filters will be loaded into the main memory when the system starts.
To check whether a key exists in a run, the Bloom filter associated with the run is first probed. An I/O is incurred only when the filter returns true, otherwise the access to the sorted run can be exempted. In contrast, range-lookup aims to extract the entries for a given key range. The process is similar to point-lookup, except that (1) Bloom filter is not applied; and (2) every sorted run is accessed.

{
\vspace{0.5mm}
\noindent\textbf{Read-write Trade-off.} 
% The read-write trade-off in an LSM-tree refers to the balance between lookup cost and update cost. By modifying its structure, the LSM-tree can enhance the performance of either lookups or updates at the expense of the other. Different growth schemes yield distinct trade-off curves, indicating that the degree of sacrifice required in the update or lookup cost to achieve a certain performance level in the other metric varies. A design with a superior read-write trade-off requires less compromise in one metric to achieve a certain performance level in the other.
The “read-write trade-off” in an LSM-tree is the balance between lookup and update costs. By changing its structure, the LSM-tree can optimize for either lookups or updates, but improving one comes at the expense of the other. Different growth schemes produce different trade-off curves, showing how much one cost must be sacrificed to improve the other. A design with a better read-write trade-off requires less compromise in one cost to achieve a certain performance level in the other.
%For instance, the curves of different methods depicted in Figure 8(a) can be viewed as a practical quantitative analysis of their respective trade-offs.
}

\section{Analysis of {\Scheme}s}
\label{section: analysis}

\setlength{\textfloatsep}{8pt}
{
\begin{algorithm}[b]
  \caption{\Horizontal}
  \label{alg: horizontal}
    \KwIn{
      Maximum number of levels $\ell$
    }
    \For{$i=1\ to\ \ell$}{
        $C_i\gets 0$
    }
    \For {each flush from buffer}{
        $C_1\gets C_1+1$;\\
        \For{$i=1\ to\ \ell-1$}{
            \If{$C_i>C_{i+1}$ }{
                Trigger a full compaction from Level $i$ to Level $i+1$;\\
                $C_{i+1}\gets C_{i+1}+1$, $C_{i}\gets 0$;\\
            }
        }
    }
\end{algorithm}
}

\noindent{\bf {\Leveling}.  }%(Algorithm~\ref{alg: log scheme}).
In the {\leveling}, the capacity of each level is fixed, with a size ratio $T$ between two adjacent levels. 
A compaction happens when the data in the level reaches the predefined capacity, subsequently triggering a merge of the data in that level to the next level. A new level will be created to accommodate the compaction of the current largest level. %If the data volume of the largest level reaches its capacity, the LSM-tree will create a new level below it, with $T$ times larger capacity than the previous largest level. 
As a result, the LSM-tree grows vertically (i.e., additional levels are gradually created) as the overall data volume expands.
%Specifically, as shown in Algorithm~\ref{alg: log scheme}, 
Given the buffer size $B$, the capacity of Level $i$ is set as $B\cdot T^i$, a compaction from Level $i$ to Level $i+1$ is triggered when the data volume at Level $i$ surpasses its capacity.

\vspace{1mm}
\noindent{\bf{\Horizontal}. }%(Algorithm~\ref{alg: horizontal}).} 
In the horizontal scheme, the number of levels is strictly fixed. Consequently, with the influx of more data, the LSM-tree broadens horizontally, meaning that each level's capacity would increase over time.
As shown in Algorithm~\ref{alg: horizontal}, let $\ell$ be the maximum number of levels, and $C_i$ be the times of compactions conducted from Level $i-1$ to Level $i$ since the previous compaction from Level $i$ to Level $i+1$. Especially, for the first level, $C_1$ is the number of flushes from memory buffer to Level 1 since the previous compaction from Level 1 to Level 2. For the largest level, $C_{\ell}$ is the number of flushes from Level $\ell-1$ to Level $\ell$. A full compaction from Level $i$ to Level $i+1$ is triggered whenever $C_i > C_{i+1}$. Intuitively, the horizontal scheme dynamically controls the data volume ratio between two adjacent levels by their relative numbers of compactions. %\dingheng{More intuitive.}

\begin{figure*}[t]
    \centering
    \vspace{0mm}
    \includegraphics[width=0.95\textwidth]{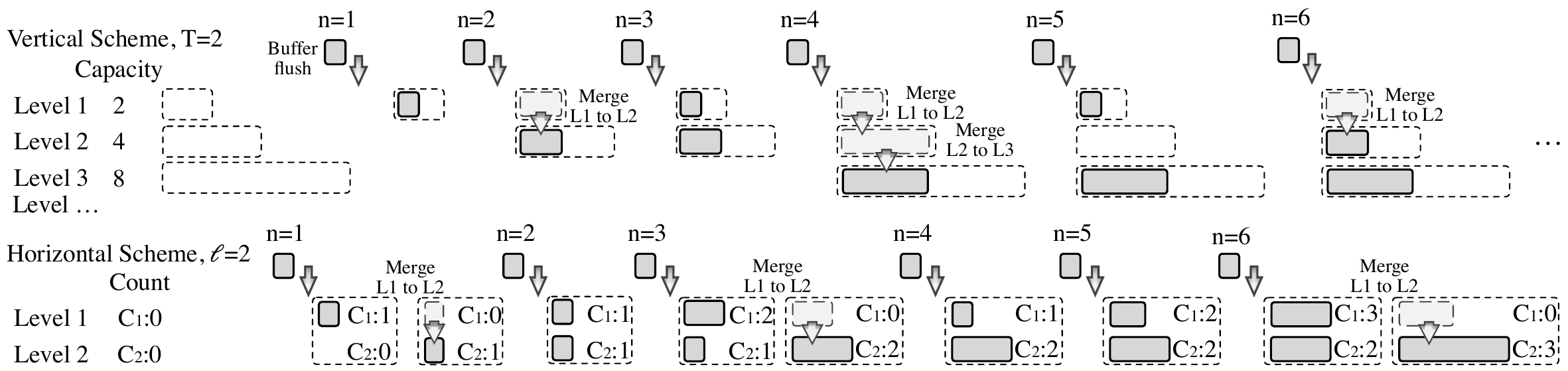}
    \vspace{-1mm}
    \caption{Running examples of the vertical and horizontal growth schemes.}
    \label{fig:scheme example}
    \vspace{0mm}
\end{figure*}

{
\vspace{1mm}
\noindent\textbf{Running Example (Figure~\ref{fig:scheme example}).} 
%An example comparing the two growth schemes is shown in Figure~\ref{fig:scheme example}. 
Let $n$ be the number of times a full memory buffer is flushed to disk. %For instance, when $t=5$, it means the buffer has been flushed to the disk five times.
%For intuitive illustration, both schemes perform full compaction in this example. 
We set size ratio $T=2$ for the vertical scheme, and let the maximum number of levels $\ell$ be $2$ in the horizontal scheme.
Initially, each level of the vertical scheme has a predetermined capacity, with the first level having a capacity that can hold 2 buffers, and the second level can hold 4 buffers. In contrast, the horizontal scheme's levels do not have a predetermined capacity. Under the vertical scheme, when the number of buffer flushes $n=1$, the LSM-tree consists solely of Level 1, accommodating data equivalent to the buffer size. When $n=2$, the data volume of Level 1 reaches its capacity after two buffer flushes, and a compaction from Level 1 is triggered, with Level 2 created as the target level. The subsequent compaction from Level 1 to Level 2 occurs at $n=4$, at which point Level 2 also reaches its capacity after two compactions from Level 1. Consequently, a new Level 3 is then established, enforcing a compaction from Level 2 to Level 3. In this way, LSM-tree under the vertical scheme continuously expands vertically. 
On the other hand, the horizontal scheme assigns each level a compaction counter, starting at 0. Each buffer flush increases the counter $C_1$ at Level 1 by one. After the first buffer flush ($n=1$), we have $C_1=1$ and $C_2=0$. Given these values, the compaction condition for Level 1 under the horizontal scheme, $C_1 > C_2$, is met, thereby triggering a compaction from Level 1 to Level 2. After this compaction, $C_1$ is reset to 0, while $C_2$ increases by one. When $n=3$, $C_1$ equals 2, which surpasses $C_2=1$, thus triggering another compaction from Level 1 to Level 2. Subsequently, $C_1 $ is reset to 0 again, while $C_2$ rises to 2. A similar process recurs at $n=6$. %As this process persists, the LSM-tree will broaden horizontally in correlation with the the growth in data volume.
}

%The right part of Figure~\ref{fig:scheme example} serves as a more general example over an extended time span.
%The horizontal axis corresponds to the number of buffer flushes $t$.
%Each vertical bar represents the size of levels on the disk, where deeper levels correspond to darker shades. Each grid signifies data entries equivalent to the memory buffer's capacity.
%We can observe that the level number for the horizontal scheme remains at two, while the vertical scheme continually creates additional levels as more data flushes in.

\vspace{1mm}
\noindent\textbf{Horizontal Scheme Achieves a Better Read-Write trade-off (a.k.a., Bentley and Saxe’s theory~\cite{bentley1980decomposable}).} 
The compaction from Level \( i \) to Level \( i+1 \) involves merge-sorting all data from Level \( i \) with existing data in Level \( i+1 \), keeping each level as a single sorted run. This process becomes increasingly costly as the size of Level \( i+1 \) grows, especially if equal amounts of data are merged each time, as is typical in the vertical scheme.
For example, as shown in Figure~\ref{fig: intuition example} (a), the vertical scheme moves a total of 60MB data from Level 1 to Level 2 in three equal compactions of 20MB, which results in compaction costs of 20MB, 40MB, and 60MB, leading to a total of 120MB of writes. Intuitively, a more efficient approach is to perform compaction more frequently when Level 2 is smaller and less frequently as it grows. By merging at 10MB, 20MB, and then 30MB, the compaction costs become 10MB, 30MB, and 60MB, reducing the total write cost to 100MB, as illustrated in Figure~\ref{fig: intuition example} (b). This is precisely the strategy employed by the horizontal scheme. 
% While this example considers only two levels, applying this intuition to multiple levels is significantly more complex. Ideally, compactions between Level \( i \) and Level \( i+1 \) should adhere to this intuitive approach as much as possible, without becoming excessively infrequent. Overly infrequent compactions would result in the data size on Level \( i \) growing too large, thereby increasing the overhead of moving data from Level \( i-1 \) to Level \( i \).
\setlength{\textfloatsep}{0pt}
\begin{figure}[t]
\vspace{0mm}
    \centering
    \setlength{\abovecaptionskip}{2mm}
    \setlength{\belowcaptionskip}{2mm}
    \hspace{-2mm}
    \includegraphics[width=0.95\textwidth]
    {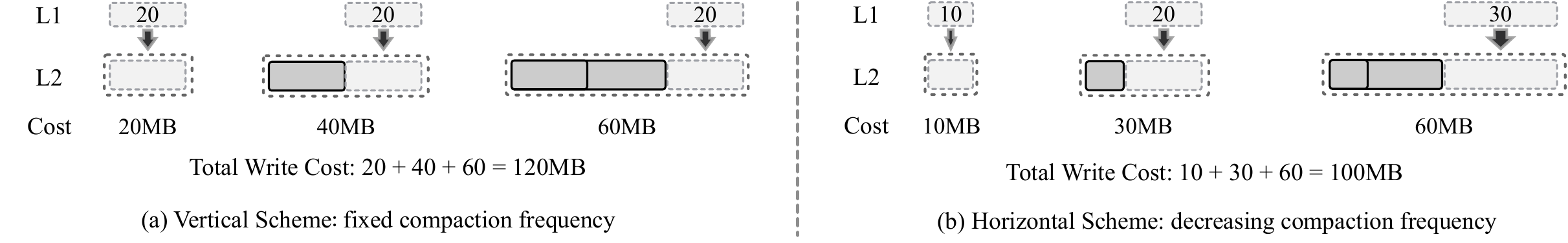}
    \vspace{-1.7mm}
    \caption{Illustrating the cost of different schemes.}
    \label{fig: intuition example}
    \vspace{-2mm}
\end{figure}

% \vspace{1mm}
% \noindent\textbf{Optimality of Horizontal Scheme.} 
In fact, Bentley and Saxe~\cite{bentley1980decomposable} (later refined by Mathieu et al.~\cite{mathieu2014bigtable} for LSM-tree context) have demonstrated the optimality of the horizontal scheme, that the horizontal scheme incurs the minimal write cost for a given number of levels. Since the read cost of an LSM-tree under the leveling merge policy is proportional to the number of levels, the horizontal scheme consistently achieves the optimal trade-off between the read and write costs.

{
\vspace{1mm}
\noindent{\bf Limitations of the Horizontal Scheme.}\label{sec:summary}
The previous analyses conclude that the traditional vertical scheme is suboptimal, while the horizontal scheme can provide better overall read-write trade-off. %The compaction timing of the horizontal scheme is more intuitive, and we have theoretically proven the optimality of its algorithm. \fan{I'm not that sure, but more intuitive seems not necessarily an explicit merit.} 
Nevertheless, as we mentioned, most recent key-value stores~\cite{rocksdb, leveldb, lakshman2010cassandra, cockroachdb, scylladb, badgerdb, wiredtiger} tend to adopt the vertical scheme. This somewhat contradicts the theoretical results presented above. 
We conjecture that this discrepancy is probably due to two practical limitations of the horizontal scheme.
%The reason for this discrepancy is that the horizontal scheme has the following two non-negligible issues.

% \vspace{1mm}
% \noindent\underline{\textit{Bounded merge policy.}} \dingheng{Refine with plain language.}
% The vertical scheme effectively supports update intensive workloads, which are common in real-world applications, by switching from the leveling merge policy to the tiering merge policy. However, for a long time, the horizontal scheme has remained bound to the leveling policy. Simply switching the merge policy in the horizontal scheme does not work as effectively. We will explore this issue in detail in Section~\ref{sec:horizontal-tiering}.

\vspace{0.5mm} 
\noindent{\underline{\textit{Limited capability in handling update-heavy workloads.}}}
Update he-avy workloads are common in practice. The vertical scheme effectively accommodates such workloads by switching to tiering merge policy. In contrast, the horizontal scheme has only been designed for the leveling merge policy. Moreover, as we will show shortly, directly applying the tiering merge policy may easily lead to sub-optimal performance. To address this limitation, we propose the \emph{horizontal-tiering} scheme, which is the first horizontal scheme algorithm that can be effectively applied under the tiering merge policy. This approach will be explained in detail in Section~\ref{sec:horizontal-tiering}.

\vspace{0.5mm}
\noindent{\underline{\textit{Gap between theory and practice.}}}
Horizontal scheme is based on {\it full compaction} granularity, whereas file-level partial compaction granularity is the standard in most industrial systems, and has demonstrated practical advantages in certain aspects of empirical system performance, including significantly lower space amplification and improved worst-case throughput. This discrepancy arises because a single full compaction merges the entire contents of one level into the next level all at once. When the sizes of the participating levels are large, this process incurs substantial space amplification (since, for consistency, the data involved in compaction must be backed up until the compaction completes) and causes write stalls (as data in preceding levels cannot be transferred while compaction is underway). In contrast, each partial compaction merges only a single file from one level with a small set of overlapping files in the next level, thereby affecting only a fraction of the data and considerably mitigating these issues. %This gap between theory and practice motivates the design of {\ourmethod}, a novel LSM-tree growth scheme that integrates both the horizontal and vertical scheme, as described in Section~\ref{sec:our-method}.
This gap between theory and practice motivates us to propose a novel design that integrates both the horizontal and vertical scheme, as described in Section~\ref{sec:our-method}.
}
% \vspace{1mm}
% \noindent
% {\bf{Remark: Gap between theory and practice. }}As we mentioned, most recent key-value stores~\cite{rocksdb, leveldb, lakshman2010cassandra, cockroachdb, scylladb, badgerdb, wiredtiger} adopt the vertical schemes. This somewhat contradicts the theoretical result given above. A worth mentioning point is that the theory developed is based on {\it full compaction} policy, whereas partial compaction policy is the standard in most industrial systems, and has demonstrated advantages in several aspects of empirical system performance. This theory-practice gap motivates us to design a brand new LSM-tree growth scheme, described in Section~\ref{sec:our-method}.

\section{The Horizontal-Tiering Scheme}
\label{sec:horizontal-tiering}

\begin{figure}[t]
    \centering
    \setlength{\abovecaptionskip}{2mm}
    \setlength{\belowcaptionskip}{2mm}
    \hspace{0mm}
    \includegraphics[width=1\textwidth]
    {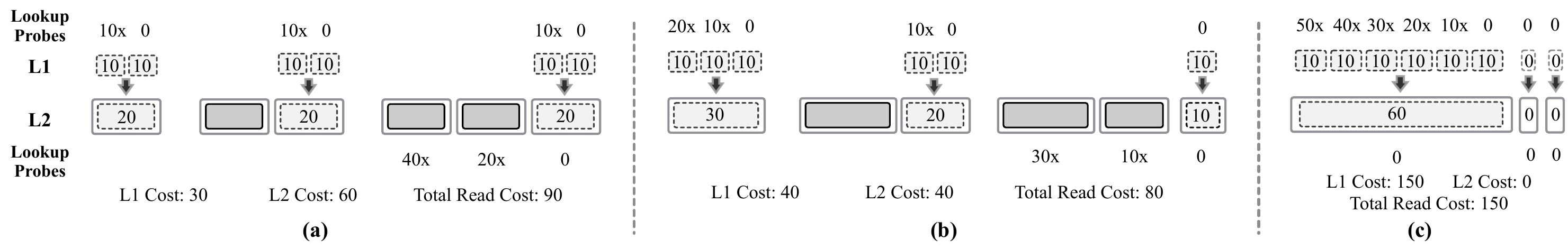}
    \vspace{-3mm}
    \caption{Different compactions lead to varying read costs.}
    \label{fig: intuition example tiering}
    \vspace{0mm}
\end{figure}

In this section, we introduce the horizontal-tiering scheme, which serves as the backbone of our proposed method, {\ourmethod}. The horizontal-tiering scheme redefines the horizontal algorithm by integrating the tiering merge policy, which governs when the data in a level is compacted into the subsequent level as a new isolated sorted run. While there are numerous compaction-timing choices in an LSM-tree, our proposed horizontal-tiering scheme employs a compaction sequence that achieves the optimal read-write performance trade-off, or precisely, achieves the minimum write cost when fixing the number of levels. 
%Since Bentley and Saxe's optimality theory was first introduced 40 years ago, 
This is the first significant extension of the horizontal growth scheme to the tiering merge policy, resulting in a much broader range of optimal read-write performance trade-offs for horizontal schemes.
Although the proof of these optimality properties is quite intricate, the algorithm of the final scheme is easy to implement, making it appealing from an engineering standpoint as well. %We will first explain the challenges of optimally combining the horizontal scheme with the tiering merge policy, followed by a formal description of the algorithm and the accompanying proofs.
% In this section, we introduce the horizontal-tiering scheme, which is a backbone of our final proposed scheme {\ourmethod}. The horizontal-tiering scheme redesigns the horizontal algorithm to incorporate the tiering merge policy. Tiering policy governs when the data in a level is compacted into the subsequent level as a new isolated sorted run, and there are enormous compaction-timing choices for an LSM-tree. Our proposed horizontal-tiering scheme incorporates an {\it optimal} compaction sequences in achieving best read-write performance trade-offs. Since Bentley
% and Saxe's optimality theory had been proposed in 1980, this is the first time to non-trivially extend this theory to the tiering merge policy, achieving a much wider optimal read-write performance trade-off for the horizontal growing schemes. Not suprisingly, the proof for the optimality properties involves a lot of sophistication, but it is worth noting that the implementation of the final scheme is surprisingly simple, making it also attractive from engineering perspective. We will first present why optimally combining horizontal scheme and tiering merge policy is challenging, followed by formal algorithmic description and proofs. %Our motivation for proposing the horizontal-tiering scheme stems from the observation that the traditional horizontal scheme is ineffective when applied directly to tiering.

% \vspace{1mm}
% \noindent\textbf{Effective Integration of Horizontal Scheme and Tiering Merge Policy is Challenging.}

\subsection{High-Level Idea}
\label{section: ht high level}
%\vspace{-1mm}
Under the tiering merge policy, we first claim that the write amplification, in the long run, is the same {\it regardless of when} we compact data from a level to its subsequent level. 
This is because, with the tiering policy, each entry incurs only a single write amplification (i.e., rewritten onto the disk once) when passing through each level, exactly at the time when it is moved onto the level. As a result, the key to optimizing overall performance lies in selecting the optimal timing for compacting data from Level $i$ into a new run at Level $i+1$, in order to minimize read amplification. %Therefore, the crux of optimizing the overall performance can be converted to choosing the best timing to compact data from Level-$i$ into a new run in Level-$(i+1)$, to {\it minimize the read amplification}. To see how challenging this problem is, below we discuss several ideas and examples. 

% One naive idea is to follow the vertical scheme, to simply use the original Algorithm~\ref{alg: horizontal}, except that the compaction moves data from Level $i$ into a new run in Level $i+1$ instead of merge-sorting them with the existing data in Level $i+1$.
% However, such a simple strategy will incur a significantly larger read amplification because the algorithm will rapidly create many runs in a level even when the level's data size is still relatively small. As the read cost is approximately proportional to the number of runs, this strategy leads to an inferior read cost at a level. This hints the importance of compaction timing on the read performance. To illustrate the complexity of this challenge, we will discuss several examples below.
A simple approach is to follow the vertical scheme by using the original Algorithm~\ref{alg: horizontal}, but instead of merge-sorting data with the existing content in Level $i+1$, the compaction would simply move data from Level $i$ into a new run in Level $i+1$. However, this naive strategy would result in significantly higher read amplification, because it would quickly create many runs in a level, even when the total data size in that level is still relatively small. Since the read cost is roughly proportional to the number of runs, this approach would lead to suboptimal read performance. This underscores the critical role of compaction timing in optimizing read efficiency. 

%To take a closer look at how the compaction timing impact the read performance, 
To further explore the complexity of this challenge, we consider three examples in Figure~\ref{fig: intuition example tiering}, where a total of 60MB of data is moved from Level 1 to Level 2 across 3 compactions as data is written into the workload. {We further assume there are $x$ lookups when 1MB of data is ingested, and each lookup incurs an I/O for every run present at that time. \footnote{For simplicity, this assumption excludes common memory optimizations typically employed in LSM-trees, such as Bloom filters and block caches. However, our conclusions remain valid even when considering these optimizations. For example, in an LSM-tree utilizing Bloom filters and block caches, if the workload is uniform, there is still a consistent probability (approximately equal to the Bloom filter’s false positive rate) that a data block retrieval will occur on each run during every lookup. In this context, the insight presented in this example continues to hold.}} %an entry size of 1KB and a balanced lookup-to-update ratio, meaning that approximately 1,000 lookups occur for every 1MB of data written. The Bloom filter's false positive rate is assumed to be 1\%, so each lookup has a 1\% probability of incurring an I/O for every run present at that time. \dingheng{Simplify.} 
Figure~\ref{fig: intuition example tiering} (a) describes an equal-frequency compaction, where a compaction from Level 1 to Level 2 happens every insertion of 20MB data. Figure~\ref{fig: intuition example tiering} (b) illustrates an increasing-frequency compaction strategy, where compactions occur after ingesting 30MB, 20MB, and 10MB of data, respectively. In contrast, Figure~\ref{fig: intuition example tiering} (c) presents an extreme case where three compactions happen consecutively after inserting 60MB of data. Now, consider the total read cost during the insertion of 60MB. We will analyze the costs at Level 1 and Level 2 separately.

At Level 1, each run is 10MB in size, and if compactions are less frequent, earlier runs remain in the level for longer, increasing the read cost. Thus, Case (c) results in significantly more I/Os at Level 1 (150 times of I/O operations), compared to Case (a) and Case (b), which incur 30 and 40 times of I/O operations, respectively. However, compacting more frequently increases the cost at Level 2, as runs are generated earlier and serve more lookups. For example, Case (a) has a higher cost at Level 2 because runs stay longer in this level compared to the other two cases. Among the three, Case (b) achieves the lowest overall I/O cost. Its first compaction delayed until 30MB of data has been written, and subsequent compactions occur more frequently, triggered after 20MB and 10MB of written data, respectively. In this manner, by the third compaction, the two earlier sorted runs in Level 2 have existed for shorter duration compared to the case in Figure~\ref{fig: intuition example tiering} (a). The first run has only undergone 30 times of lookups, while the second run has experienced just 10 times of lookups, as compared to 40 and 20 times of lookups in Case (a). This suggests that a decreasing-frequency compaction strategy may help reduce read amplification.

Building on this insight, an intuitive approach is to perform compactions to Level $i$ less frequently when it is relatively empty, because the sorted runs written at this stage will exist for a longer period and thus contribute more to read amplification. In contrast, when Level $i$ is close to full, we can add new sorted runs at shorter intervals, since they will not persist long and will therefore have a smaller impact on read amplification. We formalize the insight into developing the horizontal-tiering scheme below.
%As shown in Figure~\ref{fig: intuition example tiering} (b), we let the first compaction delayed until 30MB of data has been written, and subsequent compactions occur more frequently, triggered after 20MB and 10MB of written data, respectively. In this manner, by the third compaction, the two earlier sorted runs in Level 2 have existed for shorter durations compared to the case in Figure~\ref{fig: intuition example tiering} (a). The first run has only undergone 30x lookup probes, while the second run has experienced just 10x lookups, as compared to 40x and 20x lookups in Case (a). 

% \vspace{1mm}
% \noindent\textbf{Horizontal and Vertical Scales.}
% Within the horizontal scheme, the data volume of each level represents its horizontal scale.
% Under the {\horizontal}, the LSM-tree's total number of levels is predetermined by $\ell$.
% As for the data volume, we have the following lemma. All the missing proofs can be found in Section~\ref{sec:proofs}.

\begin{figure*}[t]
    \centering
    \vspace{0mm}
    \includegraphics[width=0.98\textwidth]{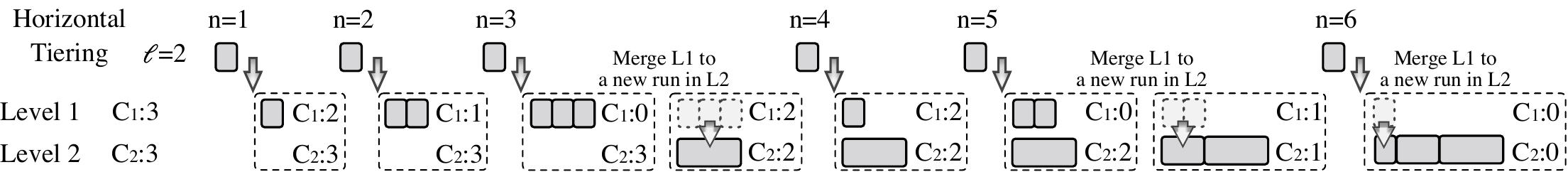}
    \vspace{-1mm}
    \caption{Running example of the horizontal-tiering scheme.}
    \label{fig:horizontal tiering example}
    \vspace{0.5mm}
\end{figure*}

\vspace{-3mm}
\subsection{Algorithms} 
% Building upon previous analysis, we present a new horizontal-tiering scheme, which extends the original horizontal scheme to better accommodate update-heavy scenarios and provides a superior read-write trade-off compared to vertical-tiering. As shown in Algorithm 3, the overall algorithm is similar to the horizontal scheme under leveling with two key differences: (1). The compaction counter initialized as a  The second difference is that the compaction threshold for each level changes from incremental to decremental. This means that compactions at each level start infrequently and become increasingly frequent as the threshold decreases.
Building upon our previous analysis, we introduce a new horizontal-tiering scheme that extends the original horizontal scheme to better accommodate update-heavy workloads, providing a superior read-write trade-off compared to vertical-tiering. As shown in Algorithm 3, the overall algorithm is similar to the default horizontal scheme under leveling, with two key differences. First, the compaction counters of all levels are initialized to an integer \( k \) instead of zero (Lines 4–5). Second, when a buffer flush or compaction occurs, the compaction counter for the target level decrements instead of incrementing. Furthermore, a level's compaction is triggered when its counter reaches zero, and after compaction, its counter is reset to the current counter value of its subsequent level (Lines 6–12).

These two differences cause compactions to become increasingly frequent as the counters decrease. 
{To illustrate this point, we present a simple example in Figure~\ref{fig:horizontal tiering example}, where \(\ell = 2\) and \(k = 3\), meaning that all levels' compaction counters are initialized to 3. With each buffer flush, the compaction counter \(C_1\) at Level 1 decreases by 1. When \(n = 3\) (after three buffer flushes), $C_1$ is reduced to zero and meets the compaction condition \(C_i = 0\), which triggers a compaction from Level 1 to Level 2. Following this compaction, \(C_2\) decreases by 1, and \(C_1\) is reset to the current value of \(C_2\), now 2.
After two additional buffer flushes, \(C_1\) decreases to 0 once again, triggering another compaction. Under the tiering merge policy, the data from Level 1 is merged into a new sorted run in Level 2. Subsequently, \(C_2\) decreases from 2 to 1, and \(C_1\) is reset to 1. The next compaction then occurs after only one buffer flush.}
As \(C_2\) decreases due to compactions, the compactions from Level 1 to Level 2 become increasingly frequent. This behavior aligns with the best-performing increasing-frequency compaction strategy depicted in Figure~\ref{fig: intuition example tiering}(b).
%which has a similar setup to the previous example in Figure~\ref{fig: scheme example} but employs the tiering merge policy.

The algorithm initializes the compaction counters for all levels to \( k \) (Lines 1–5), where \( k \) is the smallest integer satisfying the inequality
$
\frac{N}{B} \leq {{k + \ell - 1} \choose {\ell}}.
$
The reason for such setting is as follows. In running the algorithm, the compaction counters for each level will gradually decrease, until the compaction counters of all levels reach zero. Therefore, the initial value of the compaction counter must be set based on the total data size \( N \)\footnote{Specifically, $N$ represents the data size that the LSM-tree will have written by the end of Algorithm 2. In practical applications, $N$ may be set according to the capacity of the storage device or a more precise value based on user estimates. In general, as long as $N$ falls within a reasonable range, deviations in its value do not significantly affect the practical performance of Algorithm 2.}, where a larger data size corresponds to a larger initial compaction counter value \( k \). We present a novel lemma to formalize the relationship between \( k \) and the data size, which explains the initial setting of $k$.
\begin{lemma}
Under the horizontal-tiering scheme, if we initially set the compaction counters of all levels to \( k \), then after \( {{k + \ell - 1} \choose {\ell}} \) buffer flushes, the compaction counters of all levels will decrease to zero. %, and simultaneously, the compaction counters at all other levels will also be zero.
\label{lemma: binomial volume}
\end{lemma}

{
\begin{algorithm}[t]
  \caption{Horizontal-tiering Scheme}
  \label{alg: horizontal tiering}
    \KwIn{
      Maximum number of levels $\ell$, data size $N$
    }
    $k\gets 0$;\\
    \While{${{k+\ell-1}\choose{\ell}}<\frac{N}{B}$}{
        $k\gets k+1$
    }
    \For{$i=1\ to\ \ell$}{
        $C_i\gets k$
    }
    \For {each flush from buffer}{
        $C_1\gets C_1-1$;\\
        \For{$i=1\ to\ \ell-1$}{
            \If{$C_i=0$ }{
                Trigger a full compaction from Level $i$ to Level $i+1$;\\
                $C_{i+1}\gets C_{i+1}-1$;\\
                \For{$j=1\ to\ i$}{
                    $C_{j}\gets C_{i+1}$;\\
                }
            }
        }
    }
\end{algorithm}
}

\subsection{Optimality Analysis}
% \noindent\begin{problem}
% The definitions of workload and compaction remain the same as in Problem~\ref{problem: uniform leveling}. The difference is that, under tiering, a buffer flush and compaction result in the creation of a new sorted run in the target level, rather than merging it with existing runs. In this scenario, each entry incurs a fixed write amplification at each level, and our focus shifts to optimizing the overall read cost. Specifically, since the workload is uniform, we can assume there are \( r \) lookups between two consecutive buffer flushes, meaning that each buffer flush is accompanied by \(r\) lookups. Each lookup probes every run, incurring a disk I/O cost for each probe. Our objective is to minimize the total read cost of processing this workload with compaction sequence $\{S_1,S_2,...\}$.
% \label{problem: uniform tiering}
% \end{problem}
The traditional horizontal scheme is known to achieve the optimal read-write trade-off under the leveling merge policy. Our proposed horizontal-tiering algorithm generalizes this optimality to the tiering merge policy. %, providing an optimal read-write trade-off at the {\scheme} level. 
%To illustrate this, we first present the following model for LSM-trees under the tiering merge policy.
%
%Under the tiering merge policy, a buffer flush and compaction result in the creation of a new sorted run in the target level, rather than merging it with existing runs. In this scenario, each entry incurs write amplification once at each level, and the total write amplification of the LSM-tree equals its number of levels.
%Since the workload is uniform, we can assume there are a fixed number of \( r \) lookups between two consecutive buffer flushes, meaning that each buffer flush is accompanied by \(r\) lookups. Each lookup probes every run, incurring a disk I/O cost for each probe. 
%
As we mentioned, the optimality analysis can be reduced to analyzing the total read cost, when the number of levels is fixed. Next, we show that the compaction sequence described in Algorithm~\ref{alg: horizontal tiering} achieves the optimal read cost. %Different {\scheme}s can lead to varying read costs. 
%Specifically, since the workload is uniform, we can assume there are \( r \) lookups between two consecutive buffer flushes, meaning that each buffer flush is accompanied by \(r\) lookups. Each lookup probes every run, incurring a disk I/O cost for each probe. 
% Specifically, as data flux in, each buffer flush to Level 1 incurs an I/O cost of \( \frac{D_1(i)}{P} \), where \( D_1(i) \) represents the size of Level 1 after the buffer flush, and \( P \) is the disk page size.

We first state the problem setting as follows. Between consecutive buffer flushes, compactions can be performed between adjacent levels to control the number of runs of Level 1 and other levels, reducing future write costs. A compaction that starts from Level \( l_1 \) and ends at Level \( l_2 \) would merge all runs from Levels \( l_1 \), \( l_1+1 \), ..., \( l_2-1 \) into Level \( l_2 \).
Different {\scheme}s result in varying compaction timings, which ultimately lead to different read overheads. For example, if compactions are too infrequent or not performed at all, sorted runs accumulate in Level 1, making lookups increasingly costly. 
The timing of each compaction can be represented by a triple \( p = (I, l_1, l_2) \), where \( I \) denotes that the compaction occurs after the \( I \)-th buffer flush, and \( l_1 \) and \( l_2 \) indicate the starting and ending levels, respectively. This scenario can then be formulated into the following problem.

\noindent\begin{problem}
%The definitions of workload and compaction remain the same as in Problem~\ref{problem: uniform leveling}. The difference is that, 
%Under tiering, a buffer flush and compaction result in the creation of a new sorted run in the target level, rather than merging it with existing runs. In this scenario, each entry incurs a fixed write amplification at each level, and our focus shifts to optimizing the overall read cost. 
%Specifically, since the workload is uniform, we can assume there are \( r \) lookups between two consecutive buffer flushes, meaning that each buffer flush is accompanied by \(r\) lookups. Each lookup probes every run, incurring a disk I/O cost for each probe. 
Consider a workload that would incur a total of \( n \) sequential buffer flushes in the LSM-tree, which has a maximum of \( \ell \) levels. Our objective is to find the optimal compaction sequence $S^* = \{p^*_1,p^*_2,...\}$ to minimize the total read cost of processing this workload. We define this problem as $\psi(n, \ell)$.
\label{problem: uniform tiering}
\end{problem}

In a tiered LSM-tree, different {\scheme}s correspond to specific compaction sequences in Problem 1, resulting in varying costs. We give the following theorem stating that the compaction sequence employed by our horizontal-tiering algorithm incurs the minimal cost. %This conclusion is summarized in the following lemma. Here, 
Here we provide only a proof sketch, and the detailed proof is available in our technical report~\cite{techreport}.

\begin{theorem}
    When there are total data of size \( N = \binom{k + \ell - 1}{\ell} \cdot B \) written into the LSM-tree, the compaction operations in the horizontal-tiering scheme described in Algorithm~\ref{alg: horizontal tiering} is optimal for Problem 1.
\label{lemma: horizontal tiering optimality}
\end{theorem} 

\begin{proof}[Proof Sketch]
% First, we observe that Problem 1 can be decomposed into two sub-problems. Specifically, suppose the last compaction to the largest level occurs after the \( i \)-th buffer flush. We can decompose the original problem \( \psi(n, \ell) \) at this compaction point into two sub-problems: \( \psi(i, \ell-1) \) and \( \psi(n - i, \ell) \). Each of these sub-problems can further be divided into smaller problems. Therefore, to determine the optimal compaction sequence, we only need to ascertain the optimal timing for the last compaction to the largest level for each of these problems.
We can adopt a dynamic programming approach to solve Problem 1 by decomposing it into smaller subproblems. Consider the problem \( \psi(n, \ell) \), where \( n \) is the number of buffer flushes and \( \ell \) is the number of levels. Suppose that in the optimal compaction sequence \( S^* = \{p^*_1, p^*_2, \dots\} \), the first compaction to the largest level \( p^*_f \) occurs after the \( i \)-th buffer flush. We can then partition the original problem into two subproblems -- 
Subproblem \( \psi(i, \ell - 1) \) that addresses the first \( i \) buffer flushes  with the first \( \ell - 1 \) levels, and
subproblem \( \psi(n - i, \ell) \) that involves the remaining \(n - i \) buffer flushes and all \( \ell \) levels.
If the optimal compaction sequences for these subproblems are \( S_1 \) and \( S_2 \) respectively, then the optimal solution to the original problem satisfies \( S^* = \{S_1, p^*_f, S_2\} \).

We first explain why the portion of \( S^* \) preceding \( p^*_l \) must be $S_1$. Before the first compaction to the largest level in \( S^* \), the data entering the LSM-tree with the first $i$ buffer flushes remains within the first \( \ell - 1 \) levels since no compaction to Level \( \ell \) occurs before $p^*_f$. This scenario exactly matches problem \( \psi(n - i, \ell) \)
because \( S^* \) is the optimal compaction sequence that minimizes the total read cost, and the portion of \( S^* \) preceding \( p^*_f \) must correspond to the optimal solution of \( \psi(i, \ell - 1) \). Otherwise, we could replace that portion with the optimal sequence \( S_1 \), resulting in a lower total read cost, contradiction ensues. %which would contradict the assumed optimality of \( S^* \).
Similarly, the portion of \( S^* \) following \( p^*_f \) must correspond to the optimal solution of \( \psi(n - i, \ell) \), since the portion of data already moved into the largest level would not be merged with and affect any further data entering the LSM-tree. If it were not optimal, substituting it with \( S_2 \) would yield a compaction sequence with a lower read cost. %, again contradicting the optimality of \( S^* \).
By recursively applying this decomposition, we can solve Problem 1 by breaking it down into subproblems until they become trivial (e.g., when \( n = 1 \) or \( \ell = 1 \)). %For each time of decomposition, we identify a 
This dynamic programming approach constructs the optimal compaction sequence by combining the solutions of all these subproblems.

Furthermore, we can prove the following claim by mathematical induction. For any \( n \) and \( \ell \), if in the optimal compaction sequence corresponding to \( \psi(n, \ell) \), the last compaction to the largest level occurs after the \( i \)-th buffer flush, then \( i \) must satisfy
\begin{equation}
    {{m-1}\choose{\ell}}\leq n-i\leq{{m}\choose{\ell}}\ ,\ \text{and}\ \ 
    {{m-1}\choose{\ell-1}}\leq i\leq{{m}\choose{\ell-1}}
\end{equation}

where $m$ is the integer that satisfies ${{m}\choose{\ell}} \leq n \leq {{m+1}\choose{\ell}}$. 
Having identified the corresponding timing of compaction $p^*_f$, we can divide any problem into two subproblems. For each of these subproblems, we apply this equation to determine their first compaction to the largest level $p'_f$. This recursive process continues until all resulting subproblems can no longer be subdivided, that is, when either \( n = 1 \) or \( \ell = 1 \). The collection of all compactions determined through this process constitutes an optimal compaction sequence, which matches the compaction operations in Algorithm~\ref{alg: horizontal tiering}.
\end{proof}

\begin{figure*}[t]
\vspace{0mm}
    \centering
    \includegraphics[width=0.9\textwidth]{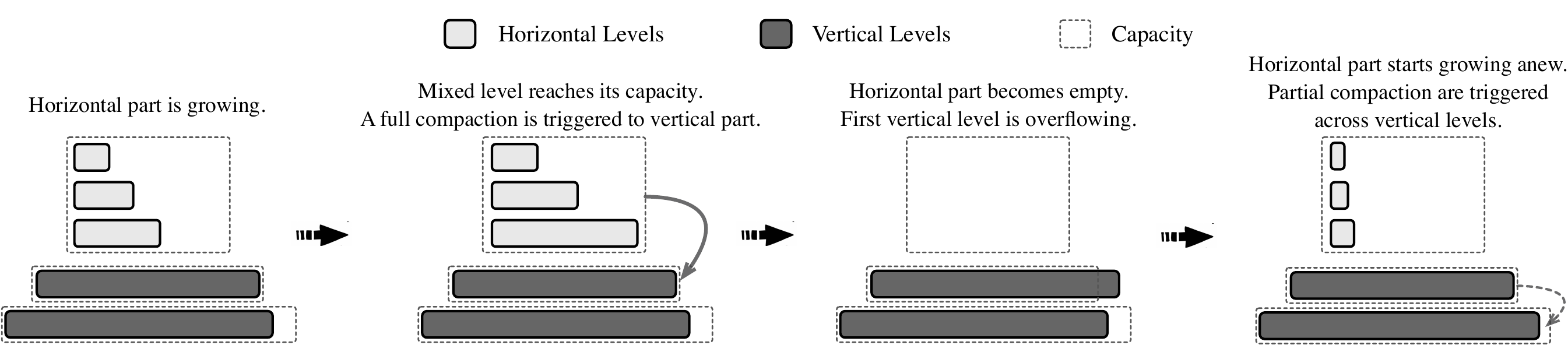}
    \vspace{-1mm}
    \caption{General process in {\ourmethod}.}
    \label{fig: hybrid example}
    \vspace{1mm}
\end{figure*}

\section{The {\ourmethod} Growth Scheme}\label{sec:our-method}

% Building upon the enhanced horizontal growing scheme introduced in the previous section, we further propose {\ourmethod}, a new LSM growing scheme that offers robust performance. 
In the previous section, we introduced the enhanced horizontal growing scheme, which maintains high efficiency under update-heavy workloads. However, it still suffers from high space amplification and write stalls due to its reliance on full compaction. This limitation prompts us to propose {\ourmethod}, a novel LSM growing scheme that combines both horizontal and vertical approaches to provide robust and comprehensive performance.

\subsection{Overall Layout}
\label{sec: hybrid scheme}
%We provide the algorithmic details of the hybrid scheme as outlined in Algorithm~\ref{alg: hybrid scheme}. 
{\ourmethod} divides the LSM-tree into two sections from top to bottom -- horizontal part and vertical part, such that the upper horizontal part utilizes the horizontal scheme for improving the overall performance, while the lower vertical part manages the largest two levels, aiming to mitigate space amplification and alleviate severe write stalls through partial compaction.

\vspace{0.5mm}
\noindent\underline{\textit{Horizontal Part.}} This part comprises the upper levels of the LSM-tree, where the vast majority of compactions occur. These levels follow the horizontal {\scheme} to achieve better efficiency. The horizontal part has a predetermined capacity, which is \( n \) times the buffer size (we will discuss how to set $n$ shortly). When the data size reaches this capacity, a full compaction is performed to merge all data into the vertical part.
The structure of the horizontal part is flexible such that it can be adjusted according to the characteristics of the workload. Specifically, it has two configurable knobs.
(1) Merge Policy: this allows switching between a leveling merge policy and a tiering merge policy. (2) Number of Levels \( \ell \): under the leveling policy, a smaller number of levels leads to better read performance; under the tiering policy, fewer levels result in better write performance.
% \begin{itemize}
% \item Merge Policy: this allows switching between a leveling merge policy and a tiering merge policy.
% \item Number of Levels \( \ell \): under the leveling policy, a smaller number of levels leads to better read performance; under the tiering policy, fewer levels result in better write performance.
% \end{itemize}
Section~\ref{sec:selftune} describes how {\ourmethod} self-tunes these parameters. %The specific algorithm for the horizontal part is similar to Algorithms 2 and 3. Additionally, it makes fine adjustments based on the preset capacity to ensure that the overall compaction sequence is optimal under this known data volume.

\vspace{0.5mm}
\noindent\underline{\textit{Vertical Part.}} The vertical part comprises the two largest levels of the LSM-tree, each possessing fixed capacities similar to those in the vertical scheme. When the horizontal part reaches its total capacity, all data is merged into the first level of the vertical part through full compaction. Once the first level becomes full, data is moved to the second level—the largest level—via partial compaction. Like the horizontal part, the vertical part can switch between leveling and tiering merge policies and includes an adjustable parameter, i.e., the size ratio \( T \) between adjacent levels in the vertical scheme.
Specifically, by default, the capacity of the first level in the vertical part is \( T \) times the total size of the horizontal part, which is \( n \cdot T \cdot B \). The second level's capacity is \( T \) times that of the first level, resulting in \( n \cdot T^2 \cdot B \). Because partial compaction between the two levels in the vertical part causes only minimal space amplification and write stalls, {\ourmethod} can roughly reduce space amplification and write stalls by a factor of \( T \) compared to an LSM-tree that employs full compaction throughout. %Therefore, increasing \( T \) can further decrease overall space amplification and write stalls but at the cost of increased write amplification in the vertical part.

%\textcolor{blue}{
%Empirically, we fix the number of levels in the vertical part to two. 
One may wonder why we empirically fix two levels in the vertical part.
The rationale behind this is that the primary purpose of the vertical part is to avoid excessive space amplification and write stalls by enabling partial compaction between the largest levels By configuring the largest two levels to use the vertical scheme and applying partial compaction between them, {\ourmethod} can significantly reduce space amplification and write stalls by roughly a factor of \( T \) compared to an LSM-tree that employs full compaction throughout. We found that practically such setting is quite sufficient and effective. {Further increasing the number of levels in the vertical part would require proportionally reducing the scale of the horizontal part, thereby diminishing its benefits.}
%This decision is based on the observation that performing full compaction on a level causes space amplification and write stalls proportional to the size of that level, whereas partial compactions consistently result in minimal space amplification and write stalls.
Note that, even though the number of levels in the vertical part is fixed, we can still trade between write amplification and space amplification by tuning the parameter \( T \). Increasing \( T \) can further decrease overall space amplification and write stalls but at the cost of increased write amplification in the vertical part.
%}
%\fan{Should we introduce {\it mixed level (level Z)} at here, which is utilized in both Figure 4 and the following analysis.}

\vspace{1mm}
\noindent
{\bf{Running Example.}}
Figure~\ref{fig: hybrid example} displays a typical process of {\ourmethod} when ingesting new data. Initially, the horizontal part is ready to grow, while the first level of the vertical part holds data volume near its full capacity. Once the horizontal part is full, it starts a compaction to the first vertical level. This process empties all the horizontal levels. Subsequently, as the data volume in the first vertical level exceeds its capacity, data begins to gradually move into the largest level via partial compactions, after which the horizontal part begins a new phase of growth.

\vspace{1mm}
\noindent\textbf{Dynamic Resizing of Horizontal Part.}
To maintain the structure of the LSM-tree as the data volume scales up, {\ourmethod} dynamically increases the parameter \( n \), which controls the capacity of the horizontal part. Specifically, $n$ is initiated with a moderate value according to the data scale. When the largest level of the vertical part reaches its capacity, we increment \( n \) by a factor of \( 1/T \) after the next full compaction clears the horizontal part, and adjust the capacity of the vertical part accordingly. This adjustment incurs no additional overhead because it is after a full clear of the horizontal part. In this way, the overall capacity of the LSM-tree under {\ourmethod} can gradually expand.
%to accommodate the increasing influx of data.

\vspace{1mm} 
\noindent
\textbf{Optimize Write Amplification by Adjusting Size Ratio.}
We can further optimize {\ourmethod}'s write amplification by adjusting the size ratio of the two largest levels in the vertical part. These two levels are special because they transition from a level employing full compaction to that with partial compaction, opening up an opportunity to optimize the overall write amplification. 

%Before delving into the technical details, let us first give some results of write amplification. 
%We begin by examining the compaction from the horizontal part into the first level of the vertical part. 
Since the vertical part employs partial compaction, the data volume in its first level consistently remains close to its capacity. Consequently, the compaction from the horizontal part to the first level of the vertical part will always incur a write amplification factor of \( T \), because the data volume in this level is maintained at \( T \) times the capacity of the horizontal part.
In contrast, the average write amplification incurred when compacting from the first level to the second level of the vertical part is only \( \frac{T + 1}{2} \). This reduction occurs because partial compaction leads to uneven data density within the level, as established in previous studies~\cite{lim2016towards}. Due to these differences, setting equal size ratios of \( T \) results in different write amplifications for compactions at the two levels of the vertical part, leading to suboptimal overall write amplification.
%Next, we demonstrate that the write amplification can be improved by tuning the size ratios between the horizontal and vertical parts.
Specifically, we set the size ratio between the horizontal part and the first level of the vertical part to \( T' \), and the size ratio between the first and second levels of the vertical part to \( {T^2}/{T'} \). This setting ensures that the capacity of the largest level in {\ourmethod} remains unchanged, thereby maintaining the overall scale of the LSM-tree. Under this setup, the write amplification at the first level of the vertical part becomes \( T' \), while the write amplification at the second level is adjusted to \( \left( {T^2}/{T'} + 1 \right) / 2 \). Note that

\begin{equation}\small
    T'+(\frac{T^2}{T'}+1)/2 \ge 2\sqrt{T'\cdot \frac{T^2}{2T'}}+\frac{1}{2}=\sqrt{2}T+\frac{1}{2}
\end{equation}
The above inequality is due to AM-GM Inequality, and the equality holds when $T'={T}/{\sqrt{2}}$.

% {\small
% \begin{algorithm}[t]
%   \caption{Hybrid Scheme}
%   \label{alg: hybrid scheme}
%     \KwIn{
%       Maximum number of levels with the horizontal scheme $Z$
%     }
    
%     \For{$i=1\ to\ Z$}{
%         $C_i\gets 0$
%     }
%     \For {each flush from buffer}{
%         $C_1\gets C_1+1$;\\
%         \For{$i=1\ to\ Z-1$}{
%             \If{$C-i>C_{i+1}$ }{
%                 Fully compact Level $i$ into Level $i+1$;\\
%                 $C_{i+1}\gets C_{i+1}+1$;\\
%                 $C_i\gets 0$;\\
%             }
%         }
        
%         \If{$i=Z$ {\bf{AND}} data size in Level $i$>$B\cdot T^i$}
%         {
%         Fully compact Level $i$ into Level $i+1$;
%         }
%         $L\gets$ current total number of levels;\\
%         \For{$i=Z+1\ to\ L$}{
%             \If{The data size in Level $i$ $> B\cdot T^i$ }{
%                 % \If{$i=Z$ }{
%                 % Fully compact Level $i$ into Level $i+1$;
%                 % }
%                 % \Else{
%                 Partially compact Level $i$ into Level $i+1$;
%                 % }
%             }
%         }
%     }
% \end{algorithm}
% }

\vspace{-2mm}
\subsection{Self-Tuning for Various Workloads}\label{sec:selftune}

\noindent
\textbf{Cost Model for Horizontal Part.}
{\ourmethod} self-tunes its horizontal part to ensure it reaches the optimal point on the read-write trade-off curve for any given workload. To this end, we have carefully analyzed the cost of each operator and proposed a cost model for the horizontal part.

\vspace{0.5mm}
\noindent\underline{\textit{Point and Range Lookup Cost.}}
A point lookup probes each sorted run from top to bottom. %Typically, a point lookup results in zero true positives at the horizontal part. This assumption is reasonable 
As the total data range is at least the size of the largest level in the vertical scheme, which is significantly larger than the total data size in the horizontal part, the probability of an arbitrary key appearing in the horizontal part is very low.
Therefore, we model the lookup cost of the horizontal-leveling scheme by multiplying the number of levels $\ell$ with the false positive rate of the Bloom filters $f$ (i.e. $R_l = \ell \cdot f$).
% \[
% R_l = \ell \cdot f
% \]
The analysis of the lookup cost for the horizontal-tiering scheme is quite intricate. Since each run exists for varying and irregular periods, they contribute differently to the overall read cost. Below we present a novel analysis results in the following lemma stating the lookup cost under horizontal-tiering scheme. 
Please refer to our technical report~\cite{techreport} for a complete proof.
%Interested readers are referred to our technical report~\cite{techreport} for a complete proof. %below, with the detailed derivation presented in the appendix.
\begin{lemma}
Let $m$ be the integer that satisfies ${{m}\choose{\ell}} \leq n \leq {{m+1}\choose{\ell}}$. Then the lookup cost of the horizontal-tiering scheme can be computed by the following equation.
\begin{equation}
R_t = \left(\ell \cdot {{m}\choose{\ell+1}} + (m-\ell+1)\cdot\left(n-{{m}\choose{\ell}}\right)\right)\cdot \frac{f}{n}
\label{equation: rt}
\vspace{-4mm}
\end{equation}
\label{lemma: rt}
\end{lemma}

%\fan{Is it viable to provide some hints to assist readers in understanding the rationale of the lemma, or figuring out what elements matter in the cost model. For instance, does m have some physical meaning. This could be important since these cost models are used for tuning.}
In terms of incurred overhead, the primary distinction between range lookups and point lookups lies in their interaction with the runs in the horizontal part of the LSM-tree. Specifically, for each existing run, point lookups typically result in a disk I/O operation only in cases of false positives of Bloom filters, whereas range lookups invariably incur a disk I/O for every run. Aside from this difference, the derivations for point lookups and range lookups are identical. Therefore, the range lookup costs under leveling and tiering satisfy \( Q_l = \frac{R_l}{f} \) and \( Q_t = \frac{R_t}{f} \), respectively.

\vspace{0.5mm}
\noindent\underline{\textit{Update Cost.}}
With tiering, each entry would only be involved in one compaction at each level.
Therefore, the I/O cost per update under the horizontal-tiering scheme can be modeled by dividing the number of levels $\ell$ by the page size in entries $P$, reminding that each disk I/O moves an entire page of entries in bulk.
% \[
% W_t = \frac{\ell}{P}
% \]
Furthermore, we employ the following lemma to calculate the update cost of the horizontal-leveling scheme.
\begin{lemma}
Let $m$ be the integer that satisfies ${{m}\choose{\ell}} \leq n \leq {{m+1}\choose{\ell}}$. Then the update cost of the horizontal-leveling scheme can be represented as the following equation.
\begin{equation}
W_l = \left(\ell \cdot {{m+1}\choose{\ell+1}} + (m+1)\cdot\left(n-{{m}\choose{\ell}}\right)-(\ell-1)\cdot n\right)\cdot \frac{1}{n\cdot P}
\end{equation}
\label{lemma: wl}
\end{lemma}

\vspace{-2mm}
\noindent
\textbf{Navigate Toward the Best Design.}
For a given design, we can weight its I/O cost of the update, point lookup, and range lookup  $W,\ R,$ and $Q$ with their respective proportion in the workload, denoted as $w,\ r,$ and $q$, to derive the average per operation I/O cost incurred by this design.
\begin{equation}
    \zeta = w\cdot W+r\cdot R+q\cdot Q
\end{equation}
We explore all possible configurations for the horizontal part of {\ourmethod}. Specifically, leveraging the convexity of the cost functions under both leveling and tiering merge policies with respect to the number of levels \( \ell \), we initiate our search from the minimum feasible value, \( \ell = 2 \) and increase \( \ell \). The optimal setting for each merge policy is identified by locating the saddle point of the cost function \( \zeta \), where the function transitions from decreasing to increasing. Finally, we select the superior design between the leveling and tiering merge policies.
To enhance the efficiency of self-tuning, we employ binary search, noting that the maximum value of \( \ell \) cannot exceed \( n \). Therefore, the complexity of the self-tuning process is only \( O\left(\log_2 n\right) \), which is efficient practically.

\vspace{-2mm}
\subsection{Optimization for Skewed Workloads}
\label{Sec: skewed}
% \noindent\textbf{Modeling Workload Skewness}
%We adopt the definition of skewness distribution in previous works~\cite{chatterjee2021cosine}, which creates a sharper truncation between hot and cold keys for ease of analysis.
% The skewed distribution consists of two distinct uniform distributions $U_c$ and $U_h$, where $U_c$ is far greater than $U_h$. Specifically, $U_c$ contains a small set of special keys which is updated frequently, and $U_j$ is a massive space of the regular keys.
% Based on this definition, we assume that the set of hot keys has a high probability of appearing in each buffer flush. 
% This results in each subsequent run containing all the keys from the hot set. The primary difference, as shown in Figure X, is that under a uniform workload, merging two runs of size \( D \) during compaction results in a run of size \( 2D \), since the keys are unlikely to overlap.  In contrast, under a skewed workload, each key from the hot set would occur in both runs, and the obsolete keys are discarded during compaction, resulting in a run size of only \( 2D - U_h\cdot E \). Therefore, under the leveling merge policy, each level grows more slowly with each compaction as the workload becomes more skewed.

%\noindent\textbf{Modeling Workload Skewness.}
As an optional technique, we can further optimize our technique by making use of skewness information. We adopt the definition of skewed distribution from previous works~\cite{chatterjee2021cosine}, which introduces a sharp distinction between hot and cold keys for analytical convenience. The skewed distribution comprises two distinct uniform distributions: \( U_h \) and \( U_c \), where \( U_h \) represents a small set of frequently updated hot keys, and \( U_c \) denotes a much larger set of regular cold keys.
We assume that the hot key set has a high probability of appearing in each buffer flush, and it is likely that each subsequent run containing all the keys from the hot set. For instance, under a uniform workload, merging two runs of size \( D \) during compaction results in a new run of size \( 2D \) because the keys are unlikely to overlap. In contrast, under a skewed workload, both runs contain the hot keys, and obsolete versions are discarded during compaction. This leads to a run size of only \( 2D - U_h \), while \( U_h  \)  redundant hot keys are eliminated. Therefore, under the leveling merge policy, each level grows more slowly with each compaction as the workload becomes more skewed.
Whereas, under the tiering merge policy, the impact of workload skewness on performance is less significant because new data flowing in each level does not merge with existing data.

%\vspace{1mm}
%\noindent\textbf{Insight Under Skewed Workload.}
We then discuss how we adapt the compaction strategy when the impact of skewness is more pronounced, i.e., under the leveling merge policy. The write amplification of a level, incurred by the compaction of its upper level, is proportional to the current size of the level. % increases as data flows in from the level above, and this increase is directly proportional to the size of the  level. 
%Performing compaction clears the data in the level, thereby reducing the write amplification of subsequent compactions. 
Under skewed workloads, the growth of levels is slower than in a uniform workload because more duplicate entries appear in a compaction, leading to a slower increase in write amplification. Consequently, we see a diminished benefit of performing compaction, which can reduce the write amplification of subsequent compactions. %, making compaction less necessary. % when a level remains relatively small. 
%Note that this insight is specific to the leveling merge policy and does not apply to the horizontal-tiering scheme.\siqiang{Will discuss the insights for tiering as well.}
% {\small
% \begin{algorithm}[t]
%   \caption{{\Horizontal} under Skewed Workload}
%   \label{alg: horizontal skew}
%     \KwIn{
%       Maximum number of levels $\ell$, Skewness factor $\alpha$
%     }
%     \For{$i=1\ to\ \ell$}{
%         $C_i\gets 0$
%     }
%     \For {each flush from buffer}{
%         $C_1\gets C_1+1$;\\
%         \For{$i=1\ to\ \ell-1$}{
%             $\delta\gets 0$;\\
%             \If{$i=1$}{
%                 $\delta\gets {\sqrt{\frac{15\alpha +1}{4-4\alpha}}-\frac{1}{2}}$
%             }
%             \If{$C_i>C_{i+1} + \delta$ }{
%                 Trigger a full compaction from Level $i$ to Level $i+1$;\\
%                 $C_{i+1}\gets C_{i+1}+1$;\\
%                 $C_{i}\gets 0$;\\
%             }
%         }
%     }
% \end{algorithm}
% }
% \vspace{1mm}
% \noindent\textbf{Adapting to Varying Skewness.}
Building on this insight, it is beneficial to align the compaction frequency with the workload skewness. %we found that the horizontal scheme requires adjustment to accommodate skewed workloads. 
Interestingly, only minor modifications to Algorithm~\ref{alg: horizontal} are necessary. Specifically, we introduce a coefficient \( \alpha = \frac{U_h}{B} \) to quantify the degree of workload skewness, where \( B \) is the buffer size. This coefficient also reflects the slowdown in the growth of the first level's size under a skewed workload, as each hot key has a high probability of appearing in each buffer flush. As the workload becomes more skewed, %the size of the hot key set \( U_h \) increases, resulting in a larger \( \alpha \).
or \( \alpha \) increases, compactions should occur less frequently since the level grows slower. %growth rate of the level sizes slows down.  
This adjustment is implemented by increasing the compaction trigger threshold of the first level based on skewness $\alpha$. While there are multiple ways to achieve the adjustment, in this paper, we change the condition in Line 6 in Algorithm~\ref{alg: horizontal} to $C_i>C_{i+1}+\delta$ when $i=1$, where %to denoted by \( \delta \), which is adjusted according to the skewness factor \( \alpha \). 
\(\delta\) is the largest integer satisfying
\begin{equation}
\vspace{-1mm}
    %\alpha \geq \binom{\delta+1}{2}\cdot(1-\alpha)
    \frac{\alpha}{1-\alpha} \geq \frac{\delta\cdot(\delta+1)}{2}.
\label{equation: delta}
\vspace{0.5mm}
\end{equation}
In Equation~\ref{equation: delta}, the term \({\alpha}/(1-\alpha)\) represents a skewness-related coefficient, reflecting the ratio of hot keys in each flush. Notably, it shows an approximately quadratic relationship with \(\delta\). Equation~\ref{equation: delta} is grounded in a deep intuition comparing the benefits and costs in adjusting \(\delta\). Due to space limitation, we refer interested readers to our technical report~\cite{techreport} for a rigorous derivation of this equation.

\subsection{{\ourmethod} as a Backbone}
\label{sec: embed}
% The core optimization of {\ourmethod}, the {\scheme}, is a fundamental component of the LSM-tree, enabling it to serve as a universal infrastructure. This allows it to be seamlessly integrated into various state-of-the-art LSM-based designs, offering further performance optimizations.

%Plentiful new LSM-based designs~\cite{dayan2017monkey,lu2017wisckey, dostoevsky2018, liu2024structural, chatterjee2021cosine, huynh2021endure, luo2020breaking} are built upon the vertical scheme. %By optimizing the {\scheme}, the foundational element of the LSM-tree architecture, {\ourmethod} serves as a universal infrastructure. This enables it to be seamlessly embedded into these state-of-the-art LSM-based designs, providing opportunities for further performance improvements.
For existing LSM-based designs using the vertical scheme, our improved {\ourmethod} can serve as a backbone to replace the vertical scheme in enhancing their system performance. %, and here we discuss detailed techniques to realize this purpose.

For example, we consider the well-known LSM-based design lazy-leveling, proposed in Dostoevsky~\cite{dostoevsky2018}. In lazy-leveling, the largest level of the LSM-tree is configured using the leveling merge policy, while the remaining levels employ tiering. We can embed {\ourmethod} into this design by replacing all the tiering levels with the horizontal part of {\ourmethod} utilizing the horizontal-tiering scheme. Specifically, suppose the original lazy-leveling design has \( L \) levels. To maintain consistency during embedding, we set the number of levels \( \ell \) in the horizontal part to \( L - 1 \) and adjust its capacity to \( B \cdot T^{L - 1} \), corresponding to the size of the largest level being replaced. In this way, we integrate {\ourmethod} into lazy-leveling without altering the original configuration.
Since the number of levels using tiering remains unchanged after the replacement, the update cost is the same as that of the original lazy-leveling design. Meanwhile, the lookup cost is improved thanks to the nice property of horizontal scheme (Theorem~\ref{lemma: horizontal tiering optimality}). %, due to the optimality of horizontal-tiering, the average costs for range and point lookups are reduced. Consequently, the overall performance of the lazy-leveling design is enhanced after embedding {\ourmethod}.

%Similarly, we can extend Moose~\cite{liu2024structural}, a state-of-the-art LSM-based structure, by embedding {\ourmethod} into it. Moose is structurally more sophisticated than Lazy-leveling, offering a larger design space that includes size ratios and merge policies, and it dynamically adjusts its settings based on the workload. Nevertheless, since all levels in Moose adhere to the vertical scheme, we can still utilize the horizontal part of {\ourmethod} to align with and replace its consecutive tiering or leveling levels within Moose's design, thereby achieving better overall performance.

\vspace{1mm}
\noindent\textbf{Introducing New Bloom Filter Layout.}
Some LSM-based designs employ the Monkey Bloom filter layout~\cite{dayan2017monkey} to improve their point lookup performance. The idea is to allocate higher bits-per-key for Bloom filters at smaller levels. % bits based on the capacity of each level. 
%The underlying rationale is that allocating a specific quantum of main memory to Bloom filters at larger levels yields marginal improvements in read efficiency, whereas the same memory allocation can more efficaciously reduce the false positive rate in smaller levels.
However, the optimization strategy rests on the assumption that the data size within every level consistently remains at that level's capacity. 
This assumption is not accurate under the full compaction granularity employed in many state-of-the-art LSM-based designs~\cite{liu2024structural, wang2024grf, dayanspooky}. Because each full compaction empties a level after merging all its data. Consequently, the level would progress from being vacant to reaching its full capacity before another compaction occurs. This problem also impedes the compatibility of the Monkey layout with {\ourmethod}.

% \begin{figure}[t]
%     \centering
%     \setlength{\abovecaptionskip}{2mm}
%     \setlength{\belowcaptionskip}{2mm}
%     \hspace{-2mm}
%     \includegraphics[width=0.48\textwidth]
%     {figures/bf_example.pdf}
%     \caption{Reallocation in the dynamic bloom filter layout.}
%     \label{fig: dynamic bf example}
% \end{figure}

% \vspace{1mm}
% \noindent\textbf{Dynamically Adjust Bloom Filters in Response to Data Volume.}

To address this problem, we propose a new dynamic Bloom filter layout, which is designed to achieve better read optimization under the horizontal scheme by more precisely modeling full compaction processes. 
Unlike the static allocation of the Bloom filter under the Monkey layout, the dynamic layout adaptively adjusts the Bloom filter's allocation in response to changes in data volume for each level and incorporates future changes in data volume into the model's considerations. However, directly adjusting Bloom filter's allocation requires fetching the original data indexed by the Bloom filter from disk into memory in order to reconstruct the Bloom filter based on the new bits per key, which incurs significant additional overhead. To circumvent this issue, we never directly adjust Bloom filter's allocation under the dynamic layout. Each time a full compaction occurs between two levels and data is loaded into memory, we simultaneously reallocate and reconstruct the Bloom filter associated with those two levels.

\section{Additional Related Works}
%\vspace{1mm}
\noindent\textbf{LSM-tree-based Designs.}
%The LSM-tree is a prevalent foundation for the storage backend of key-value databases in commercial applications~\cite{rocksdb,polardb,dgraph,tidb,cockroachdb}.
%As such, it has been the focus of extensive research aimed at its enhancement. 
%A wide range of LSM-tree-based designs have been proposed with various focuses, 
Recent LSM-tree-based optimizations include exploring innovative compaction policies and routines ~\cite{dayan2017monkey, dostoevsky2018, idreos2019designcontinuum, chatterjee2021cosine, huynh2021endure, dayan2019log, sears2012blsm, dayanspooky, sarkar2020lethe, raju2017pebblesdb}, improving Bloom filters~\cite{dayan2021chucky, zhang2018elasticbf, zhu2021reducing, dayan2017monkey} and range filters~\cite{luo2020rosetta,knorr2022proteus,zhang2018surf}, as well as harnessing novel hardware for performance gains~\cite{ahmad2015compaction, huang2019x, vinccon2018noftl, zhang2020fpga, wang2014efficient, yu2022treeline}.
Techniques have also been investigated for key-value separation~\cite{chan2018hashkv, lu2017wisckey}, managing in-memory hot entries~\cite{balmau2017triad}, enhancing concurrency~\cite{golan2015scaling, shetty2013building}, and prefix indexing~\cite{luo2020rosetta, zhang2018surf, wu2015lsmtrie}. 
Furthermore, some efforts have been dedicated to practical performance enhancements by reducing tail latency~\cite{balmau2019silk, luo2019performance, sears2012blsm}, exploiting workload characteristics~\cite{absalyamov2018lightweight, ren2017slimdb, yang2020leaper}, and refining the memory distribution~\cite{bortnikov2018accordion, kim2020robust, luo2020breaking}.
To the best of our knowledge, most of these works are either based on the vertical growth scheme or are orthogonal to the growth scheme, as the vertical scheme is typically considered the default in LSM-trees. %In contrast, {\ourmethod} primarily focuses on developing an efficient growth scheme for LSM-trees and can potentially serve as a foundational optimization that can be integrated with many of the aforementioned LSM-tree-based designs.

\vspace{1mm}
\noindent\textbf{Studies on Horizontal {\Scheme}.}
%Several studies~\cite{mathieu2014bigtable, mathieu2021competitive, mao2019experimental, mao2021comparison, mao2023comparison} have discussed the horizontal scheme. Mathieu {\it et al.}~\cite{mathieu2021competitive}  analyzed and optimized the compaction policy under the horizontal scheme from the perspective of competitive ratios. Mao {\it et al.}~\cite{mao2019experimental,mao2021comparison,mao2023comparison}, on the other hand, undertook a series of experimental studies to compare the vertical and horizontal schemes under regular and spatial workloads. While these works are intriguing, their focus still differs significantly from this paper.
Over the past decade, several studies~\cite{mathieu2014bigtable, mathieu2021competitive, mao2019experimental, mao2021comparison, mao2023comparison} have recognized the advantages of the horizontal scheme and have focused on its analysis and application. Mathieu \textit{et al.}~\cite{mathieu2021competitive} analyzed and optimized the horizontal growth scheme in BigTable from the perspective of competitive ratios. In contrast, Mao \textit{et al.}~\cite{mao2019experimental, mao2021comparison, mao2023comparison} conducted a series of experimental studies comparing the vertical and horizontal schemes under regular and spatial workloads. Although these studies are insightful, their efforts primarily concentrated on theoretical and experimental investigations of the existing horizontal growth scheme using the leveling merge policy, without proposing new designs.
In this paper, we further enhance the applicability of the horizontal scheme through a novel design. Specifically, we introduce the horizontal-tiering scheme, which extends the horizontal scheme to the tiering merge policy, and present {\ourmethod}, an enhanced growth scheme that embodies the strengths of both the vertical and horizontal schemes.

%\dingheng{to do: rephrasing the related works and adding more details.}

\vspace{-2mm}
\section{Evaluation}
 
%This section presents experimental studies that compare {\ourmethod} with other baselines on various aspects. 
By default, our experiments are conducted on a server with Intel(R) Gold 6326 CPU @ 2.90GHz processor, 256GB DDR4 memory, and 1TB NVMe SSD, running 64-bit Ubuntu 20.04.4 LTS on an ext4 partition.

% \vspace{1mm}
% \noindent\textbf{Implementation.} 

\vspace{1mm}
\noindent\textbf{Baselines and Implementation.}
We assess {\ourmethod} (abbr. VRN) against ten baselines. 
The implementation of {\ourmethod} and all the baselines are built on Facebook's RocksDB \cite{rocksdb}, a popular LSM-tree-based key-value storage system and also been employed by many previous LSM-tree studies~\cite{dayan2017monkey, dostoevsky2018, ren2017slimdb, lu2017wisckey, dayan2019log, dayanspooky, luo2020breaking, zhang2018surf, yu2022treeline, chan2018hashkv}.

The first four baselines are all possible combinations of compaction schemes and merge policies discussed in our technical section: vertical-leveling (abbr. VT-Level-Part), vertical-tiering (abbr. VT-Tier-Part), horizontal-leveling (abbr. HR-Level), and horizontal-tiering (abbr. HR-Tier). By default, the vertical schemes employ partial compaction, while the horizontal schemes use full compaction. 
This configuration makes vertical-leveling equivalent to the default setup of RocksDB. For the other three baselines, we extended RocksDB's code to implement them.
{
 We also included two additional baselines derived from the default RocksDB (i.e., Vertical-Level-Part): “Universal” and “RocksDB-Tuned”. They share most of the same settings with Vertical-Level-Part but are extended in two key ways. Specifically, Universal enables RocksDB’s universal compaction~\cite{universal}, an age-based compaction scheme aligned with the tiering merge policy, whereas RocksDB-Tuned applies additional tunings to certain RocksDB parameters specific to the vertical growth scheme.
~\footnote{Specifically, we enable "level\_compaction\_dynamic\_level\_bytes", set "max\_bytes\_for\\\_level\_multiplier" to 10, and configure "compaction\_pri" to "kOldestSmallestSeqFirst".}}
Furthermore, although not commonly used in practical applications, designs based on the vertical scheme can adopt full compaction to sacrifice fine-grained file-level compaction for improved performance. This leads to two additional baselines: vertical-leveling with full compaction (abbr. VT-Level-Full) and vertical-tiering with full compaction (abbr. VT-Tier-Full).
Additionally, we compare two designs based on {\ourmethod} but without self-tuning. Instead, they use a fixed leveling design and a fixed tiering design with \( \ell = 2 \) in the horizontal part, referred to as VRN-Level and VRN-Tier, respectively.

\vspace{1mm}
\noindent\textbf{Experimental Setting.} {
For all methods, we retain most orthogonal parameters as RocksDB’s default settings, which are known to be developer-optimized for general SSD-based applications~\cite{tuningguide}. Furthermore, following previous works~\cite{dostoevsky2018,huynh2021endure}, we enable direct I/O by setting the parameters “use\_direct\_io\_for\_flush\_and\_compaction” and “use\_direct\_reads” to true in RocksDB. We configure the memory buffer size to 2 MB by setting “write\_buffer\_size” and “target\_file\_size\_base” to 2,097,152, and set the LSM-tree size ratio \(T=6\) by assigning “max\_bytes\_for\_level\_multiplier” to 6. %We employ RocksDB’s default Bloom filter policies with a bits-per-key setting of 5. 
%For all methods, we let most of orthogonal parameters adhere to RocksDB’s default settings, which are already a set of developer-optimized configurations considered practical for general SSD-based applications~\cite{tuningguide}. Furthermore, following previous works \cite{dostoevsky2018,huynh2021endure}, we enable direct I/Os for read and write (by setting parameters "use\_direct\_io\_for\_flush\_and\_compaction" and "    use\_direct\_reads" as true). The memory buffer size is configured at 2MB (by setting "write\_buffer\_size" and "target\_file\_size\_base" as 2097152), the size ratio of LSM-tree $T$ at 6 (by setting "max\_bytes\_for\_level\_multiplier" as 6). And we employ RocksDB’s default Bloom filter policies with a bits-per-key setting of 5. %Unless otherwise mentioned, we process a workload of 50 million operations. 
}
We evaluate the impact of different Bloom filter settings, varying the bits-per-key from 4 to 20 (Figure~\ref{fig:large memory exp}), while using a default value of 5 in other experiments, as this aligns with the latest RocksDB default. 
We also analyze the impact of block cache size, varying from 16MB to 64GB (Figure~\ref{fig:large memory exp} (e)), with a default of either 32MB or 64GB in other experiments. Specifically, Figure~\ref{exp: throughput}, Figure~\ref{fig:large memory exp} (a) and (d), Figure~\ref{exp: tradeoff and embed} (b) and (c), and Figure~\ref{fig: large scale exp} (a) use a 32MB block cache, while Figure~\ref{fig:large memory exp} (b) and (c), Figure~\ref{exp: tradeoff and embed} (d) and (e), along with Figure~\ref{fig: large scale exp} (b) use a 64GB block cache. These two configurations are chosen because a 32MB block cache aligns with the default setting of the latest RocksDB version, while a 64GB block cache represents scenarios with abundant memory resources. Evaluating both settings reveals different performance rankings among the baselines, providing a more comprehensive understanding of system performance.  
%By default, the block cache size follows RocksDB's default setting of 32MB, as block cache is an orthogonal optimization to growth schemes. 
%To demonstrate the performance of {\ourmethod} and each baseline in different scenarios, 
We evaluate each under various workloads with different lookup and update ratios; we initially bulk-load $10^7$ key-value entries of 1KB into the LSM-tree. Each entry is generated through the YCSB benchmark with either uniform or Zipfian distribution, comprising a 128-byte key and an 896-byte value.
For testing, we generate and process a workload with 4,000,000 operations, each can be a lookup or an update, and report the overall average throughput and the worst-case throughput of processing the workload. To measure the worst-case throughput, we maintain a time window containing the most recent 100,000 processed operations and record the lowest throughput observed within this window, which is then reported as the worst-case throughput.
To measure space amplification, we record the peak disk space occupied by RocksDB's storage directory during runtime. By subtracting the data size from the total space used by the LSM-tree, and then dividing the difference by the size of unique entries, we derive the additional space amplification per entry. %We utilize RocksDB's internal metric APIs to gather necessary statistics. 

\begin{figure}[t]
\vspace{0mm}
    \centering  
    \hspace{-2mm}
    \begin{minipage}[h]{.68\linewidth}
    \centering  
    \vspace{2mm}
    \begin{minipage}[h]{1.38\linewidth}
    \hspace{3mm}
    \vspace{1mm}
    \includegraphics[width=1.00\linewidth]{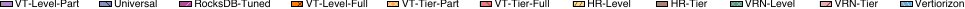}
    \end{minipage}
    \vspace{0.5mm}
    \begin{overpic}[width=1.0\linewidth]{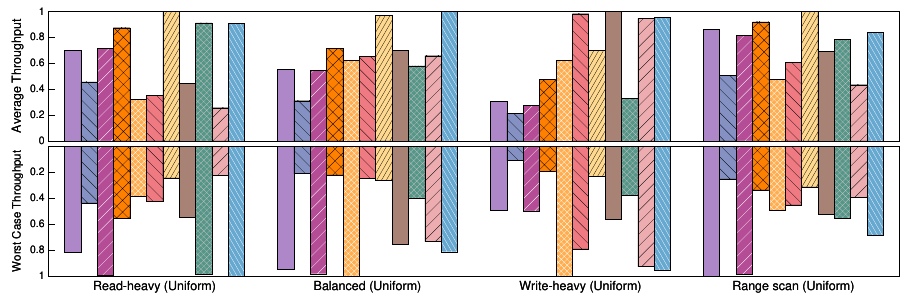}\put(42, -2){\color{black}{\sffamily\footnotesize{(a) Uniform queries}}}
    \end{overpic}\\
    \vspace{2mm}
    \begin{overpic}[width=1.0\linewidth]{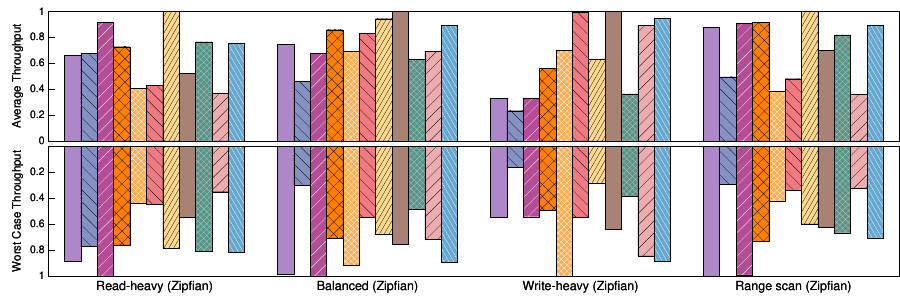}\put(42, -2){\color{black}{{\sffamily\footnotesize(b) Zipfian queries}}}
    \end{overpic}\\
    \end{minipage}
    \hspace{2mm}
    \begin{minipage}[h]{.3\linewidth}
    \begin{minipage}{1\linewidth}
    \vspace{-2mm}
        \begin{overpic}[width=0.95\linewidth]{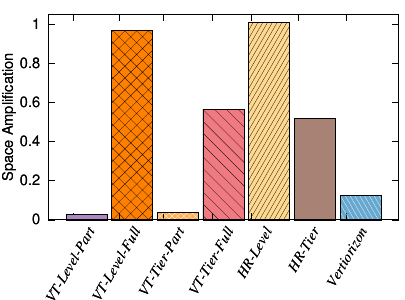}\put(50, 65){\color{black}{\sffamily\footnotesize{(c)}}}
    \end{overpic}
    \vspace{1mm}
    \end{minipage}
    \scalebox{0.68}{
    \small
\addtolength{\tabcolsep}{-4pt}
\renewcommand\arraystretch{1.1}
\sffamily
\begin{tabular}{|cc|cccc|cc|cc|c|}
\hline
\multicolumn{2}{|c|}{} & \multicolumn{4}{c|}{Vertical} & \multicolumn{2}{c|}{Horizon} & \multicolumn{2}{c|}{Hybrid} &  \\ \cline{3-10}
\multicolumn{2}{|c|}{\multirow{-2}{*}{\begin{tabular}[c]{@{}c@{}}\textbf{(d)} \\ Metric\end{tabular}}} & LP & LF & TP & TF & {\ \ L\ \ } & T & L & T & \multirow{-2}{*}{\begin{tabular}[c]{@{}c@{}}VRN \end{tabular}} \\ \hline
\multicolumn{1}{|c|}{} & RH & \cellcolor[HTML]{FFFE65}5 & \cellcolor[HTML]{FFFE65}4 & \cellcolor[HTML]{FD6864}8 & \cellcolor[HTML]{FE996B}7 & \cellcolor[HTML]{67FD9A}1 & \cellcolor[HTML]{FE996B}6 & \cellcolor[HTML]{67FD9A}2 & \cellcolor[HTML]{FD6864}9 & \cellcolor[HTML]{67FD9A}2 \\ \cline{2-2}
\multicolumn{1}{|c|}{} & B & \cellcolor[HTML]{FD6864}9 & \cellcolor[HTML]{67FD9A}3 & \cellcolor[HTML]{FE996B}7 & \cellcolor[HTML]{FFFE65}5 & \cellcolor[HTML]{67FD9A}2 & \cellcolor[HTML]{FFFE65}4 & \cellcolor[HTML]{FD6864}8 & \cellcolor[HTML]{FFFE65}5 & \cellcolor[HTML]{67FD9A}1 \\ \cline{2-2}
\multicolumn{1}{|c|}{} & WH & \cellcolor[HTML]{FD6864}9 & \cellcolor[HTML]{FE996B}7 & \cellcolor[HTML]{FE996B}6 & \cellcolor[HTML]{67FD9A}1 & \cellcolor[HTML]{FFFE65}5 & \cellcolor[HTML]{67FD9A}1 & \cellcolor[HTML]{FD6864}8 & \cellcolor[HTML]{67FD9A}3 & \cellcolor[HTML]{67FD9A}3 \\ \cline{2-2}
\multicolumn{1}{|c|}{\multirow{-4}{*}{\rotatebox[origin=c]{90}{Average}}} & RS & \cellcolor[HTML]{67FD9A}3 & \cellcolor[HTML]{67FD9A}2 & \cellcolor[HTML]{FD6864}8 & \cellcolor[HTML]{FE996B}7 & \cellcolor[HTML]{67FD9A}1 & \cellcolor[HTML]{FE996B}6 & \cellcolor[HTML]{FFFE65}5 & \cellcolor[HTML]{FD6864}9 & \cellcolor[HTML]{67FD9A}3 \\ \cline{1-2}
\multicolumn{1}{|c|}{} & RH & \cellcolor[HTML]{67FD9A}3 & \cellcolor[HTML]{FCFF2F}4 & \cellcolor[HTML]{FE996B}7 & \cellcolor[HTML]{FE996B}6 & \cellcolor[HTML]{FD6864}8 & \cellcolor[HTML]{FFFE65}4 & \cellcolor[HTML]{67FD9A}1 & \cellcolor[HTML]{FE996B}8 & \cellcolor[HTML]{67FD9A}1 \\ \cline{2-2}
\multicolumn{1}{|c|}{} & B & \cellcolor[HTML]{67FD9A}2 & \cellcolor[HTML]{FD6864}9 & \cellcolor[HTML]{67FD9A}1 & \cellcolor[HTML]{FD6864}8 & \cellcolor[HTML]{FE996B}7 & \cellcolor[HTML]{FFFE65}4 & \cellcolor[HTML]{FE996B}6 & \cellcolor[HTML]{FFFE65}5 & \cellcolor[HTML]{67FD9A}3 \\ \cline{2-2}
\multicolumn{1}{|c|}{} & WH & \cellcolor[HTML]{FE996B}6 & \cellcolor[HTML]{FD6864}9 & \cellcolor[HTML]{67FD9A}1 & \cellcolor[HTML]{FFFE65}4 & \cellcolor[HTML]{FD6864}8 & \cellcolor[HTML]{FFFE65}5 & \cellcolor[HTML]{FE996B}7 & \cellcolor[HTML]{67FD9A}3 & \cellcolor[HTML]{67FD9A}2 \\ \cline{2-2}
\multicolumn{1}{|c|}{\multirow{-4}{*}{\rotatebox[origin=c]{90}{Worst Case}}} & RS & \cellcolor[HTML]{67FD9A}1 & \cellcolor[HTML]{FD6864}8 & \cellcolor[HTML]{FFFE65}5 & \cellcolor[HTML]{FE996B}6 & \cellcolor[HTML]{FD6864}8 & \cellcolor[HTML]{FFFE65}4 & \cellcolor[HTML]{67FD9A}3 & \cellcolor[HTML]{FE996B}7 & \cellcolor[HTML]{67FD9A}2 \\ \cline{1-2}
\multicolumn{2}{|c|}{Space} & \cellcolor[HTML]{67FD9A}1 & \cellcolor[HTML]{FD6864}8 & \cellcolor[HTML]{67FD9A}2 & \cellcolor[HTML]{FE996B}7 & \cellcolor[HTML]{FD6864}9 & \cellcolor[HTML]{FE996B}6 & \cellcolor[HTML]{67FD9A}3 & \cellcolor[HTML]{67FD9A}3 & \cellcolor[HTML]{67FD9A}3 \\ \hline
\multicolumn{2}{|c|}{Avg.} & 4.3 & 5.1 & 5 & 5.7 & 5.4 & 4.4 & 4.8 & 5.8 & 2.2 \\ \hline
\end{tabular}
}
\vspace{-7mm}
\end{minipage}
\hspace{-2mm}

\vspace{0mm}
\caption{{\ourmethod} shows desirable performance across all scenarios and metrics. Table (d) presents the ranking of each method under various metrics.
%and Table (e) presents the latency per operation (in μs) of {\ourmethod} and all baselines for each workload. 
(For methods, "L" means leveling, "T" means tiering, "P" means partial, "F" means full, and "VRN" is {\ourmethod}. For workloads, "RH" is read-heavy, "B" is balanced, "WH" is write-heavy,  and "RS" is range scan.)  }
\vspace{0mm}
\label{exp: throughput}
\end{figure}
% \begin{figure}[t]
%     \centering
%     \vspace{-3.5mm}
%     \hspace{2mm}
%     \includegraphics[width=0.8\linewidth]{figures/filter_cache_legend.eps}\\
%     \vspace{0mm}
%     \includegraphics[width=0.49\linewidth]{figures/filter_exp_average.eps}
%     \includegraphics[width=0.49\linewidth]{figures/cache_exp_average.eps}
%     \vspace{-3.7mm}
%     \captionsetup{font=small}
%     \caption{\color{teal}{Comparative results as bloom filter and cache size increase.}}
%     \vspace{0mm}
%     \label{fig:filter and cache exp}
% \end{figure}

% \begin{figure*}[t]
%     \centering
%     \vspace{-4mm}
%     \hspace{-1mm}
%     \begin{overpic}[width=0.48\linewidth]{figures/throughput_large_scale_flat.pdf}\put(33, -2){\color{black}{{(a) Moderate Memory Setting}}}
%     \end{overpic}
%      \begin{overpic}[width=0.48\linewidth]{figures/throughput_large_scale_64_flat.pdf}\put(38, -2){\color{teal}{(b) Large Memory Setting}}
%     \end{overpic}
%     \vspace{-1mm}
%     %\caption{\color{blue} Comparison of {\ourmethod} and baselines under very large data scale.} 
%     %\captionsetup{font=small}
%     \caption{Evaluation of {\ourmethod} under very large data scale.} 
%     \label{fig: large scale exp}
%     \vspace{-4.5mm}
% \end{figure*}

\vspace{1mm}
\noindent\textbf{{\ourmethod} achieves desirable performance across all scenarios and dimensions.}
In Figure~\ref{exp: throughput}~(a) and (b), we compare {\ourmethod} against all baselines in terms of average throughput and worst-case throughput across various YCSB workloads under both uniform and Zipfian distributions, with y-axis standardized to facilitate comparison. The workloads include balanced workload (50\% updates and 50\% point lookups), write-heavy workload (90\% updates and 10\% point lookups), read-heavy workload (10\% updates and 90\% point lookups), and range scan workload (25\% range lookups and 75\% updates). 
Overall, the results under the uniform and Zipfian distributions are similar. 
From Figure~\ref{exp: throughput} we have the following discussions. %We will begin by examining the results under different uniform workloads.
%We will begin by examining the results under the uniform distribution.

\vspace{0.5mm}
\noindent\underline{\textit{Results on read-heavy workload.}} 
(1) Leveling-based approaches consistently achieve better overall performance compared to tiering-based methods, because tiering-based methods tend to suffer from high read amplification, resulting in lower throughput in read-heavy scenarios.
(2) Under the vertical scheme, the full compaction method (e.g., VT-Level-Full) 
delivers superior average throughput compared to its corresponding partial compaction method (e.g., VT-Level-Part). This is because partial compaction introduces additional write amplification. However, the worst-case throughput of VT-Level-Full is lower than VT-Level-Part, because full compaction leads to significant write stalls when merging large levels. We also note that the worst-case throughput is also influenced by overall performance. For instance, VT-Tier-Part has the lowest average throughput among all vertical baselines, even though it incurs minimal write stalls with partial compaction. %, its worst-case throughput is still low, because the poor average throughput negatively impacts its worst-case performance.
(3) Baselines under the horizontal scheme demonstrate higher average throughput than their vertical scheme counterparts with the same merge policy, owing to the optimal read-write trade-off of the horizontal scheme. For example, the average throughput of HR-Level is notably better than VT-Level-Full. Nevertheless, because these horizontal methods rely on full compaction, they suffer from reduced worst-case throughput due to write stalls.
(4) {\ourmethod} excels in both average and worst-case throughput by employing a hybrid layout that delivers strong read and write performance while alleviating write stalls. This design enables {\ourmethod} to achieve strong average throughput and also maintain the best worst-case throughput among all methods.

\vspace{0.5mm}
\noindent\underline{\textit{Results on write-heavy workload.}}
(1) All tiering-based methods exhibit desired overall performance, with HR-Tier achieving the highest average throughput and VT-Tier-Partical achieving the best worst-case throughput. In contrast, all leveling-based methods struggle in this read-heavy scenario.
(2) Similar to our observations under the read-heavy workload, both VT-Tier-Full and HR-Tier experience lower worst-case throughput compared to partial or hybrid designs like VT-Tier-Part and {\ourmethod}, due to write stalls caused by full compaction. However, the relative decline in worst-case throughput is less pronounced than that observed among leveling-based methods under the read-heavy workload. This is because, under the tiering merge policy, compaction data is not merged with existing data in the next level, resulting in smaller compacted run sizes compared to leveling and thus reduced write stalls.
(3) Although VT-Tier-Part achieves the best worst-case throughput among tiering-based methods, it suffers from significantly lower average throughput—only 60\% of HR-Tier's average throughput. This degradation is attributed to partial compaction under tiering, which introduces higher write amplification and increased read amplification, as runs at each level are not cleared after compaction and continue to accumulate.
(4) Most leveling-based baselines perform poorly overall. Among these, HR-Level shows relatively better performance due to the superior read-write trade-off of the horizontal scheme.
(5) Meanwhile, {\ourmethod} achieves satisfactory performance in both average throughput and worst-case throughput on write-heavy workload, through self-tuning to a tiering-based design similar to VRN-Tier.

% Under write-heavy workloads, VT-Tier-Full, HR-Tier, and {\ourmethod} all exhibit comparable average throughput, with HR-Tier achieving the highest average throughput. Similar to our findings on the read-heavy workload, both VT-Tier-Full and HR-Tier experience worse worst-case throughput compared to {\ourmethod} due to write stalls caused by full compaction. In contrast, under the tiering merge policy, since compaction data is not merged with existing data in the next level, the amount of data involved in compaction is smaller compared to leveling, resulting in reduced write stalls. Consequently, the gap in worst-case throughput between VT-Tier-Full, HR-Tier, and {\ourmethod} is smaller than the gap observed among VT-Level-Full, HR-Level, and {\ourmethod} under the read-heavy workload, although it remains noticeable.
% Meanwhile, VT-Tier-Part achieves the best worst-case throughput but exhibits significantly lower average throughput compared to the other tiering baselines. For instance, it attains only 60\% of the average throughput of HR-Tier. This performance degradation is attributed to partial compaction under tiering, which not only introduces higher write amplification but also increases read amplification. This is because runs at each level are not cleared after compaction from that level and continue to exist. Under write-heavy workloads, most leveling-based baselines perform poorly overall; among these, HR-Level shows relatively better performance due to the horizontal scheme's superior read-write trade-off.

\vspace{0.5mm}
\noindent\underline{\textit{Results on balanced workload.}} Under the balanced workload, {\ourmethod} achieves the highest average throughput by automatically tuning to a design that optimally suits balanced read-write demands. Simultaneously, {\ourmethod} maintains satisfactory worst-case throughput. Most of the baselines exhibit markedly lower average throughput, except for HR-Full, which attains average throughput close to the best. This performance is attributed to the optimal read-write trade-off provided by the horizontal scheme and its settings being well-positioned on the trade-off curve for the balanced workload.
VT-Level-Part and VT-Tier-Part achieve the best worst-case throughput due to their use of partial compaction. However, this also leads to their average throughput being relatively poor among the baselines. Meanwhile, the worst-case throughput of several baselines employing the tiering merge policy is unsatisfactory. Finally, although VRT-Level and VRT-Tier adopt the same hybrid layout as {\ourmethod}, their fixed settings prioritize either read or write performance at the expense of the other, resulting in mediocre performance under the balanced workload.

\vspace{0.5mm}
\noindent\underline{\textit{Results on range scan workload.}} %Although range lookups constitute only 25\% of the operations in the range scan workload, 
The performance trends are similar to those observed under the read-heavy workload, because a range lookup is sustaintially more costly than a point lookup. %incurs an I/O operation on every run in the LSM-tree, leading to significantly higher read costs compared to point lookups.  
Notably, we observe a more substantial variation in worst-case throughput across methods, as the heavier updates in this workload further exacerbates performance variability caused by write stalls.

\vspace{0.5mm}
\noindent\underline{\textit{Assessing advanced variations of RocksDB.}} 
From Figures~\ref{exp: throughput}(a) and (b), Universal consistently underperforms under all settings due to its simplistic trigger condition for universal compactions, which limits algorithmic efficiency. RocksDB-tuned typically surpasses Vertical-Level-Part (the default RocksDB configuration), especially under read-heavy workloads. However, its average throughput still struggles in update-intensive scenarios, as it is essentially based on the vertical scheme and the leveling merge policy.

% \begin{figure}[t]
%     \centering
%     \vspace{-4mm}
%     \hspace{-1mm}
    
%     \includegraphics[width=0.48\textwidth]{figures/throughput_large_scale_flat.pdf}
%     \vspace{-7.2mm}
%     %\caption{\color{blue} Comparison of {\ourmethod} and baselines under very large data scale.} 
%     \captionsetup{font=small}
%     \caption{Evaluation of {\ourmethod} under very large data scale.} 
%     \label{fig: large scale exp}
%     \vspace{0mm}
% \end{figure}

\vspace{0.5mm}
\noindent\underline{\textit{Comparison of space amplification.}} Figure~\ref{exp: throughput}~(c) presents the space amplification under the balanced workload with uniform distribution. Among all the baselines, the two full-compaction methods, VT-Level-Full and HR-Level, entail the highest space amplification. In contrast, VT-Tier-Full and HR-Tier, which also use full compaction but with the tiering policy, have a lower space amplification. This reduction occurs because, under the tiering merge policy, compactions to a level do not merge with existing data, which would temporarily duplicate during the merge sort process. Compared methods with full compaction, the methods with partial compaction often entail much lower space amplification, as exemplified by VT-Level-Part and VT-Tier. % The default RocksDB configuration, represented by VT-Level, employs partial compaction by default. Therefore, VT-Level's space amplification is very low, and we can also observe a similar trend in VT-Tier.
Meanwhile, {\ourmethod} has a quite desirable space amplification, as it adopts the vertical scheme with partial compaction in the largest levels.% -- the dominating levels for space amplifications.

%Meanwhile, {\ourmethod}, VRN-Level, and VRN-Tier all adopt the hybrid layout proposed in {\ourmethod}, which limits space amplification by applying the vertical scheme to the largest two levels. Consequently, in this experiment, they incur relatively small space amplification, lower than all baselines except for VT-Level-Part and VT-Full.

\vspace{0.5mm}
\noindent\underline{\textit{Overall performance ranking.}}
Figure~\ref{exp: throughput}~(d) presents the ranking of each method across various performance metrics, including average and worst-case throughput across all uniform workloads, as well as space amplification. 
Overall, the results demonstrate that {\ourmethod} consistently delivers strong performance across all metrics,  whereas each baseline falls short in at least one dimension of performance.

\begin{figure*}[t]
    \centering
    \vspace{0mm}
    \begin{minipage}[h]{.81\linewidth}
    \hspace{5mm}\includegraphics[width=0.95\linewidth]{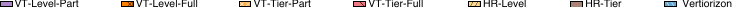}
    \includegraphics[width=1.00\linewidth]{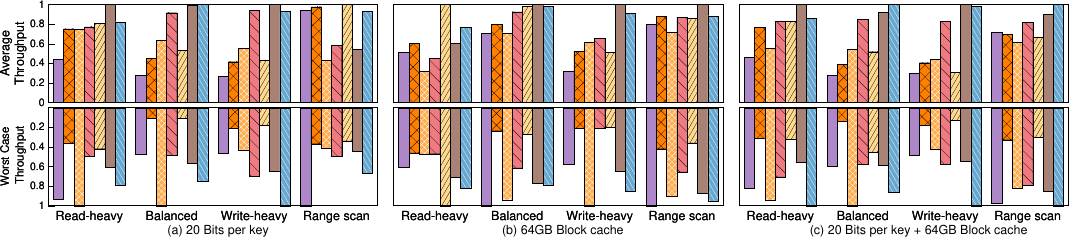}
    \includegraphics[width=0.995\linewidth]{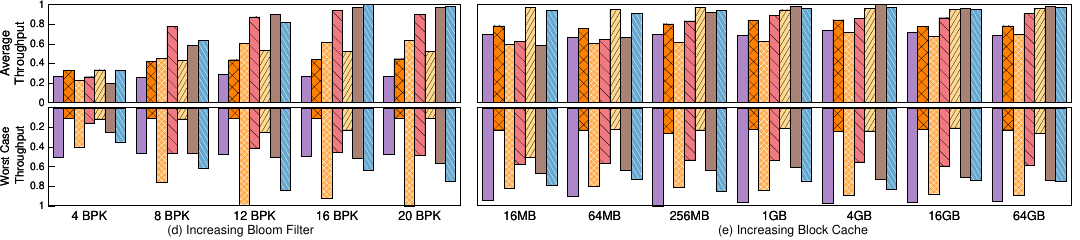}
    \end{minipage}
% \usepackage{multirow}
\hspace{-2.5mm}
\begin{minipage}[h]{.19\linewidth}
\vspace{-2mm}
\captionsetup{singlelinecheck = false, justification=raggedright, margin={0mm, 1.5mm}, labelfont={tiny}, textfont={tiny},}
\captionof{table}{Rankings of each method across different metrics for each case in Fig.~\ref{fig:large memory exp} (rounded to nearest tenth).}\label{tab:large}
\vspace{-3mm}
\scalebox{0.48}{
\addtolength{\tabcolsep}{-3.4pt}
\renewcommand\arraystretch{1.39}
\sffamily
\begin{tabular}{|cc|cccc|cc|c|}
\hline
\multicolumn{2}{|l|}{} & \multicolumn{4}{c|}{Vertical} & \multicolumn{2}{c|}{Horizon} &  \\ \cline{3-8}
\multicolumn{2}{|c|}{\multirow{-2}{*}{Metrics}} & LP & LF & TP & TF & \ \ L\ \ \  & T & \multirow{-2}{*}{VRN} \\ \hline
\multicolumn{1}{|c|}{} & (a) & \cellcolor[HTML]{FD9A52}6.0  & \cellcolor[HTML]{FCE638}4.5  & \cellcolor[HTML]{FDB44A}5.5  & \cellcolor[HTML]{E3FF41}3.5  & \cellcolor[HTML]{D9FF48}3.3  & \cellcolor[HTML]{B2FE65}2.5  & \cellcolor[HTML]{99FE76}2.0 \\ \cline{2-2}
\multicolumn{1}{|c|}{} & (b) & \cellcolor[HTML]{FD9A52}6.0  & \cellcolor[HTML]{FCFF2F}4.0  & \cellcolor[HTML]{FD8B58}6.3  & \cellcolor[HTML]{FCFF2F}4.0  & \cellcolor[HTML]{E3FF41}3.5  & \cellcolor[HTML]{8FFE7D}1.8  & \cellcolor[HTML]{8FFE7D}1.8 \\ \cline{2-2}
\multicolumn{1}{|c|}{} & (c) & \cellcolor[HTML]{FD8B58}6.3  & \cellcolor[HTML]{FCBE46}5.3  & \cellcolor[HTML]{FCBE46}5.3  & \cellcolor[HTML]{D9FF48}3.3  & \cellcolor[HTML]{FCCD41}5.0  & \cellcolor[HTML]{80FD88}1.5  & \cellcolor[HTML]{80FD88}1.5 \\ \cline{2-2}
\multicolumn{1}{|c|}{} & (d) & \cellcolor[HTML]{FD8659}6.4  & \cellcolor[HTML]{FCC344}5.2  & \cellcolor[HTML]{FCEB36}4.4  & \cellcolor[HTML]{C0FE5A}2.8  & \cellcolor[HTML]{FCF533}4.2  & \cellcolor[HTML]{CAFE53}3.0  & \cellcolor[HTML]{8FFE7D}1.8 \\ \cline{2-2}
\multicolumn{1}{|c|}{\multirow{-5}{*}{\rotatebox[origin=c]{90}{Average}}} & (e) & \cellcolor[HTML]{FDA94D}5.7  & \cellcolor[HTML]{FCEB36}4.4  & \cellcolor[HTML]{FD6D62}6.9  & \cellcolor[HTML]{FCEB36}4.4  & \cellcolor[HTML]{9EFE73}2.1  & \cellcolor[HTML]{BBFE5D}2.7  & \cellcolor[HTML]{9EFE73}2.1 \\ \cline{1-2}
\multicolumn{1}{|c|}{} & (a) & \cellcolor[HTML]{C0FE5A}2.8  & \cellcolor[HTML]{FD815B}6.5  & \cellcolor[HTML]{CAFE53}3.0  & \cellcolor[HTML]{E3FF41}3.5  & \cellcolor[HTML]{FD8B58}6.3  & \cellcolor[HTML]{FCF034}4.3  & \cellcolor[HTML]{8FFE7D}1.8 \\ \cline{2-2}
\multicolumn{1}{|c|}{} & (b) & \cellcolor[HTML]{B2FE65}2.5  & \cellcolor[HTML]{FD8B58}6.3  & \cellcolor[HTML]{C0FE5A}2.8  & \cellcolor[HTML]{FDB44A}5.5  & \cellcolor[HTML]{FCCD41}5.0  & \cellcolor[HTML]{E3FF41}3.5  & \cellcolor[HTML]{A8FE6C}2.3 \\ \cline{2-2}
\multicolumn{1}{|c|}{} & (c) & \cellcolor[HTML]{A8FE6C}2.3  & \cellcolor[HTML]{FD815B}6.5  & \cellcolor[HTML]{CAFE53}3.0  & \cellcolor[HTML]{FCFF2F}4.0  & \cellcolor[HTML]{FD815B}6.5  & \cellcolor[HTML]{E3FF41}3.5  & \cellcolor[HTML]{80FD88}1.5 \\ \cline{2-2}
\multicolumn{1}{|c|}{} & (d) & \cellcolor[HTML]{F2FF36}3.8  & \cellcolor[HTML]{FD6864}7.0  & \cellcolor[HTML]{71FD93}1.2  & \cellcolor[HTML]{FCE13A}4.6  & \cellcolor[HTML]{FD9A52}6.0  & \cellcolor[HTML]{DEFF44}3.4  & \cellcolor[HTML]{A3FE6F}2.2 \\ \cline{2-2}
\multicolumn{1}{|c|}{\multirow{-5}{*}{\rotatebox[origin=c]{90}{Worst Case}}} & (e) & \cellcolor[HTML]{67FD9A}1.0  & \cellcolor[HTML]{FD8659}6.4  & \cellcolor[HTML]{9EFE73}2.1  & \cellcolor[HTML]{FCCD41}5.0  & \cellcolor[HTML]{FD7C5D}6.6  & \cellcolor[HTML]{FCFF2F}4.0  & \cellcolor[HTML]{C5FE56}2.9 \\ \hline
\multicolumn{2}{|c|}{Avg.} & \cellcolor[HTML]{FCF034}4.3  & \cellcolor[HTML]{FDAE4B}5.6  & \cellcolor[HTML]{FCFA31}4.1  & \cellcolor[HTML]{F7FF33}3.9  & \cellcolor[HTML]{FCD23F}4.9  & \cellcolor[HTML]{CAFE53}3.0  & \cellcolor[HTML]{99FE76}2.0 \\ \hline
\end{tabular}
}
\end{minipage}
    \vspace{-3mm}
    %\captionsetup{font=small}
    \caption{The performance with varied bits per key and block cache size.}
    \vspace{0mm}
    \label{fig:large memory exp}
\end{figure*}

\vspace{1mm}
\noindent
%{\ourmethod} maintains robust and well-rounded performance with practical in-memory optimizations.
\textbf{The impact of large block cache and Bloom filters.}
Figure~\ref{fig:large memory exp} presents the throughput of each approach with large Bloom filters (up to 20 bits per key) and block cache (up to 64GB). The workloads for (d) and (e) follow a balanced distribution, consisting of 50\% updates and 50\% point lookups, uniformly distributed. Aside from these varying parameters, all other settings follow the settings in Figure~\ref{exp: throughput}. To examine whether {\ourmethod} maintains its superiority under large Bloom filters and block cache, we present the corresponding performance rankings in Table 3, aligned with each subfigure in Figure~\ref{fig:large memory exp}. %remain consistent with the previous experiment in Figure~\ref{exp: throughput}. 

\hypertarget{pg: cache and filter}{
Several key observations emerge when comparing these rankings with Figure~\ref{exp: throughput}. {\bf First}, {\ourmethod} remains highly competitive across all rankings and consistently achieves the highest overall performance. While our approach may not
always be the best for a specific workload, this overall ranking
highlights its {\it robustness} across various scenarios. In contrast, every other method exhibits at least one scenario where its effectiveness diminishes. 
{\bf Second}, when ample memory resources are available, all approaches are expected to achieve better performance, resulting in closer performance across them. As shown in Figure~\ref{fig:large memory exp} (d), a higher bits-per-key setting reduces read costs, leading to increased throughput. Similarly, Figure~\ref{fig:large memory exp} (e) shows a slight performance improvement when the block cache size is large. This improvement is primarily driven by a higher cache hit rate. However, since frequent compactions tend to invalidate the block cache, the benefit remains limited even with a larger cache size. {\bf Third}, since large Bloom filters and block caches generally enhance point-lookup performance, methods that prioritize write optimization (e.g., tiering methods) tend to improve their rankings. Notably, HR-Tier now ranks second in average throughput, making it more comparable to {\ourmethod}. However, it is important to note that HR-Tier is also a novel approach introduced in this paper. %{\bf Finally}, methods that employ full compactions (i.e., VT-Level-Full and VT-Tier-Full) remain less favorable in worst-case throughput. This is because major compactions significantly degrade worst-case performance, and large block caches provide little benefit in mitigating compaction overhead. %Third, the best baseline is more competitive in memory-ample case. For example, in Figure~\ref{exp: throughput}, the best baseline is 2.2 lower ranked, whereas in Figure~\ref{fig:large memory exp}, the second best approach 
}

\begin{figure}[t]
    \centering
    \vspace{0mm}
    \begin{minipage}[h]{1.1\linewidth}
     \hspace{7mm}
    \begin{overpic}[width=0.4\linewidth]{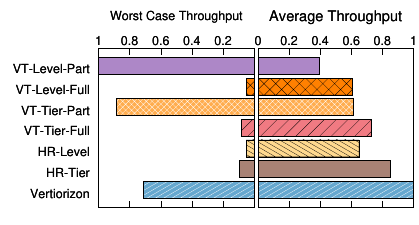}\put(30, -2.2){\color{black}{\small{(a) Moderate Memory Setting}}}
    \end{overpic}
    \hspace{-12mm}
     \begin{overpic}[width=0.4\linewidth]{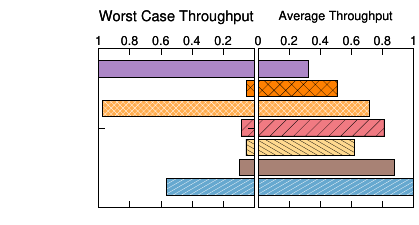}\put(33, -2.2){\color{black}{\small(b) Large Memory Setting}}
    \end{overpic}
    \end{minipage}
    \vspace{0mm}
    %\caption{\color{blue} Comparison of {\ourmethod} and baselines under very large data scale.} 
    \captionsetup{font=small}
    \caption{Evaluation of {\ourmethod} under very large data scale.} 
    \label{fig: large scale exp}
    \vspace{0mm}
\end{figure}

\vspace{1mm}
\noindent\textbf{Performance under larger datasets.}
In Figure~\ref{fig: large scale exp}, we compare our method with the baselines under significantly larger memory buffer size and data size, aiming to offer a comparative analysis that reflects scenarios of the big data era. As such large-scale experiments have high demand on disk space, we use an alternate server with a larger disk capacity, specifically featuring a 13th Gen Intel(R) Core(TM) i9-13900K CPU @ 4.0GHz processor, 128GB DDR4 main memory, and a 2TB NVMe SSD. In this experiment, most configuration settings align with precedent experiments in Figure~\ref{exp: throughput}. Except for that for all methods, the buffer size is increased to 64 MB, and the total preloaded data is increased to 500 GB, corresponding to 500,000,000 key-value entries.   We then run a workload of 200,000,000 operations to evaluate each method’s performance. Because loading such a large amount of data is time-consuming, we select a relatively general workload setting to conduct this experiment: the read-write balanced workload (50\% updates and 50\% point lookups) with a uniform key distribution. We use metrics similar in Figure~\ref{exp: throughput}~(a), including average throughput and worst-case throughput during workload processing. Figure~\ref{fig: large scale exp}~(a) uses a moderate memory configuration with 32MB block cache and 10 bits per key Bloom filter, while (b) employs a larger memory configuration with 64GB block cache and 20 bits per key Bloom filter. %Additionally, we also report the absolute latency for processing the entire workload for each method. 
%The high space amplification of the full compaction baselines can result in disk usage exceeding 1TB. Therefore, we conducted this experiment on an alternate server with a larger disk capacity, specifically featuring a 13th Gen Intel(R) Core(TM) i9-13900K CPU @ 4.0GHz processor, 128GB DDR4 main memory, and a 2TB NVMe SSD.

We observe that when the data scale becomes very large, {\ourmethod} excels in average throughput. In particular, its average throughput surpasses that of baselines such as HR-Level and HR-Tier, which had previously performed well under smaller-scale settings. The underlying reason is that these baselines rely exclusively on full compaction, leading to drastically severe write stalls at very large data scales and causing their worst-case throughput to deteriorate significantly. In contrast, by employing a hybrid growth scheme, {\ourmethod} maintains a satisfactory worst-case throughput, trailing behind methods that fully adopt partial compaction (e.g., VT-Level-Part and VT-Level-Full). Nevertheless, because partial compaction is overall less efficient, these methods post only mediocre results in terms of average throughput.

\begin{figure}[t]
\centering
\vspace{0mm}
\hspace{-4mm}
\begin{minipage}[h]{0.35\linewidth}
\vspace{-0.5mm}
\begin{overpic}[width=1\linewidth]{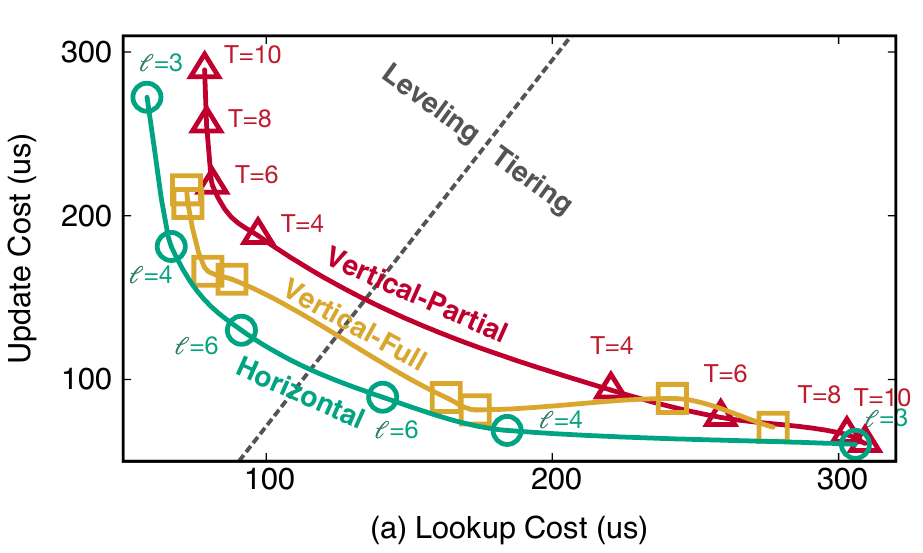}\put(78, 51){}%\color{black}{\sffamily\small(a)}}
\end{overpic}
\end{minipage}
\hspace{-2mm}
\begin{minipage}[h]{0.65\linewidth}
\vspace{-1mm}
\includegraphics[width=1.05\linewidth]{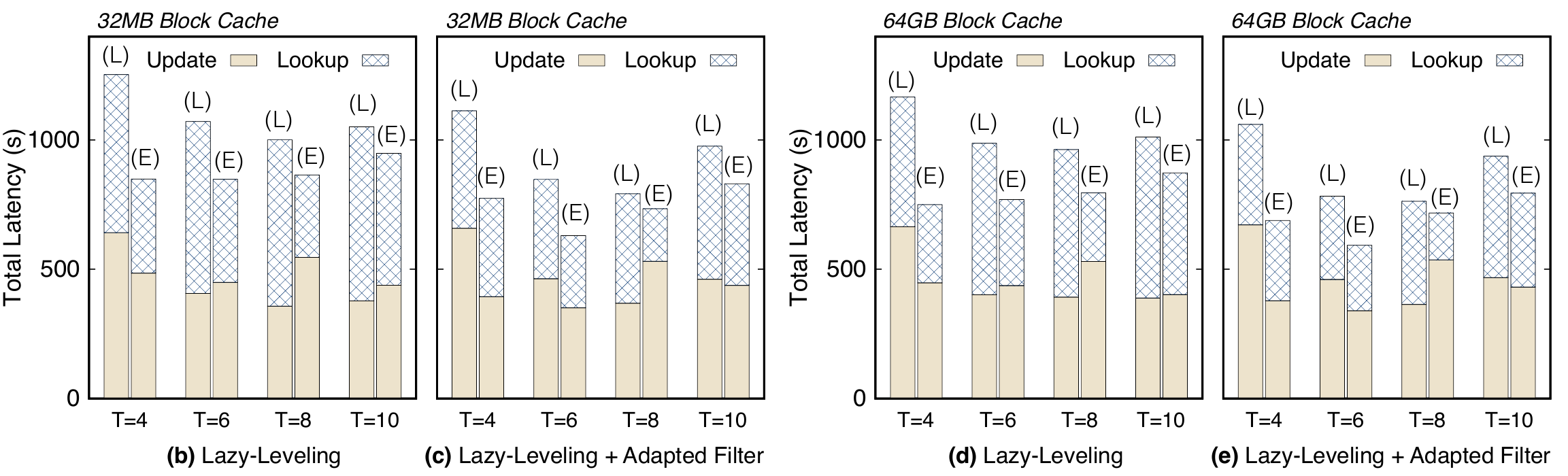}
\end{minipage}
% \begin{minipage}[h]{0.52\linewidth}
% \vspace{-1mm}
% \begin{overpic}[width=1\linewidth]{figures/embedding.pdf}\put(15, 34){\color{black}{\large\textbf{(b)}}}\put(66, 34){\color{black}{\large\textbf{(c)}}}
% \end{overpic}
% \end{minipage}
% \begin{minipage}[h]{0.26\linewidth}
% \vspace{-3mm}
% \includegraphics[width=1\linewidth]{figures/embed.eps}
% \end{minipage}
\vspace{-1.5mm}
\captionsetup{font=small}
\caption{(a) compares the trade-off between vertical scheme and horizontal scheme. (b)~$\sim$~~(e) shows that {\ourmethod} can be applied to existing designs to enhance their performance. ("L" means the lazy-leveling design and "E" means lazy-leveling embedded with {\ourmethod}.) }
\vspace{1.5mm}
\label{exp: tradeoff and embed}
\end{figure}

\vspace{1mm}
\noindent\textbf{From horizontal-leveling to horizontal-tiering: enhanced read-write trade-off curve.} 
In Figure~\ref{exp: tradeoff and embed}~(a), we compare the read-write trade-off curves of the vertical and horizontal schemes. In the vertical scheme, the trade-off between read and write performance is controlled by adjusting the size ratio \( T \) and the merge policy. The horizontal scheme's trade-off, conversely, is regulated by varying the number of levels \( \ell \) and the merge policy, and is enhanced by our horizontal-tiering algorithm under tiering merge, To evaluate these schemes comprehensively, we assess ten different designs under each compaction scheme with varying settings.
Specifically, for the vertical scheme, we test both partial compaction and full compaction designs with adjacency ratios \( T = 4, 6, 8, 10 \) under both leveling and tiering merge policies. Similarly, for the horizontal scheme, we evaluate designs with the number of levels \( \ell = 3, 4, 6\) under both merge policies. We evaluate them on a workload consisting of 20 million operations with equal amounts of updates and point lookups.
For each design, we plot a point on Figure~\ref{exp: tradeoff and embed}~(a) based on its per-lookup and per-update latency, where the horizontal axis represents the read cost and the vertical axis represents the write cost.\footnote{{The figure shows that all methods achieve relatively low latencies in both lookups and updates, on the order of hundreds of microseconds. This is primarily attributed to the inherent efficiency of the LSM-tree data structure and the high data transfer bandwidth of the SSDs used in our experimental setup. Notably, previous studies~\cite{huynh2021endure, dostoevsky2018} have also reported that, under similar experimental conditions, RocksDB demonstrates operation latencies within the hundred-microsecond range.}} By connecting all the points corresponding to each compaction scheme, we obtain the read-write trade-off curve for that scheme in practical LSM-tree implementations.

The result shows that our proposed horizontal-tiering scheme significantly extends the trade-off curve of the original horizontal scheme, thereby completely enveloping that of both vertical-partial and vertical-full designs. The enhanced horizontal scheme's trade-off dominates the vertical scheme's, with the horizontal scheme's curve consistently closer to the origin, indicating lower overall costs. %Even when employing full compaction, the vertical scheme does not outperform the horizontal scheme. 
This aligns with our earlier theoretical analysis, as the compaction sequence under the horizontal scheme is always optimal. 
%For any vertical-leveling design, we can always identify a better horizontal-leveling design that achieves the same read cost with a lower write cost. Similarly, for any vertical-tiering design, there is a corresponding horizontal-tiering design that offers the same write cost with a reduced read cost. Given the optimality of the horizontal scheme, 
The trade-off curve presented for the horizontal scheme effectively represents the Pareto frontier for LSM-trees in terms of compaction schemes.
% \textcolor{blue}{
% \dingheng{Remove this part if space is not enough.}
% In addition, the horizontal scheme offers smoother and more flexible tuning potential. Specifically, on the trade-off curve, designs of the vertical scheme are concentrated in the upper left and lower right corners, resulting in few potential designs in the middle ground between read cost and write cost (between the leveling and tiering designs with \(T = 4\), only designs with \(T = 2\) or \(T = 3\) exist). In contrast, the horizontal scheme's designs are more evenly distributed, with numerous potential designs occupying the middle ground of the trade-off curve. These designs typically do not exhibit significant weaknesses in either read or write performance, making them more desirable for practical applications.
% }

\vspace{1mm}
\noindent\textbf{{\ourmethod} enhances existing designs as a fundamental infrastructure.}
We embedded {\ourmethod} into various lazy-leveling designs and compared their performance with the original configurations. Each lazy-leveling design employed different size ratios \( T \), specifically scaling across the values 4, 6, 8, and 10. We evaluated these designs using a balanced workload comprising 4 million operations, with an equal split of 50\% updates and 50\% point lookups.
Figure~\ref{exp: tradeoff and embed}~(b) and~(d) present the total latency for each lazy-leveling design, as well as the individual read and write latencies under 32MB and 64GB block cache configurations, respectively. In this figure, the upper portion of each bar represents the read latency, while the lower portion corresponds to the write latency. The results demonstrate that embedding {\ourmethod} into any lazy-leveling design consistently improves performance. This enhancement is primarily reflected in the reduction of read latency. The improvement arises because our embedding replaces the original tiering levels with the horizontal part of {\ourmethod} utilizing the horizontal-tiering scheme, all while maintaining the same number of levels. This approach effectively reduces the average read amplification during runtime without increasing the write amplification.
Furthermore, we also tested lazy-leveling designs with adapted filter layout~\cite{dayan2017monkey}. As shown in Figure~\ref{exp: tradeoff and embed}~(c) and~(e), the results indicate that after incorporating Bloom filter optimizations, embedding {\ourmethod} as a backbone continues to enhance the performance of lazy-leveling designs. 

\section{Conclusion}
% {\color{blue}
% The growth scheme plays a critical role in determining how LSM-trees expand and maintain their hierarchical structure as they handle the continuous influx of new data entries.
% In this paper, we revisit two established growth schemes for LSM-trees: the horizontal scheme and the vertical scheme. Through a comprehensive analysis of their respective strengths and limitations, we revoke the recognition of the horizontal scheme and present an enhanced horizontal scheme by extending its applicability to accommodate the tiering merge policy. Building on this foundation, we propose {\ourmethod}, a self-designing growth scheme that leverages the individual strengths of both horizontal and vertical schemes. 
% We demonstrate the practical benefits of the enhanced horizontal scheme and {\ourmethod} by embedding these schemes into RocksDB and conducting extensive experiments. Notably, {\ourmethod} achieves robust performance across a wide range of scenarios and evaluation metrics. These results highlight the potential of {\ourmethod} to serve as a foundational infrastructure for improving general LSM-tree-based systems.
% }

In this paper, we revisit two established growth schemes for LSM-trees: the horizontal scheme and the vertical scheme. Through a comprehensive analysis of their respective strengths and limitations, we revoke the recognition of the horizontal scheme and present an enhanced horizontal scheme by extending its applicability to accommodate the tiering merge policy. Building on this foundation, we propose {\ourmethod}, a self-designing growth scheme that leverages the individual strengths of both horizontal and vertical schemes. 
We demonstrate the practical benefits of the enhanced horizontal scheme and {\ourmethod} by embedding these schemes into RocksDB and conducting extensive experiments. Notably, {\ourmethod} achieves robust performance across a wide range of scenarios and metrics. These results highlight the potential of {\ourmethod} to serve as a foundational infrastructure for improving general LSM-tree-based systems.

%This paper revisits two growth schemes of LSM-trees, namely, the horizontal scheme and vertical scheme. By analyzing the cost of the schemes, we revoke the recognition of the horizontal scheme and present an enhanced horizontal scheme by extending it to combine with the tiering merge policy. On top of this, we further present {\ourmethod}, a self-designing growth scheme that harnesses the individual strengths of both schemes. We embed our schemes into RocksDB to demonstrate their superiority.

\received{October 2024}
\received[revised]{January 2025}
\received[accepted]{February 2025}

\bibliographystyle{ACM-Reference-Format}
\bibliography{sample-base}

\vspace{10mm}
\section{Appendix}\label{sec:proofs}
\subsection{Proof of Lemma~\ref{lemma: binomial volume}}

\begin{proof}
We prove the lemma by induction on the number of levels, \(\ell\).

\paragraph{Base Case (\(\ell = 1\)).}  
When there is only one level, the compaction counter decreases by one with each buffer flush. Thus, starting from an initial counter value of \(k\), exactly \(k\) buffer flushes will reduce the counter to zero. Moreover, note that
\[
\binom{k+1-1}{1} = \binom{k}{1} = k.
\]
Hence, the lemma holds when \(\ell=1\).

\paragraph{Inductive Step.}  
Assume that for some \(\ell' \ge 1\) and for all integers \(k\), the lemma holds; that is, the total number of buffer flushes required to reduce the compaction counters to zero is 
\[
\binom{k+\ell'-1}{\ell'}.
\]
We now consider the case \(\ell = \ell' + 1\).

For \(\ell = \ell' + 1\) and any \(k\), observe that a compaction from Level \(\ell'\) to Level \(\ell'+1\) is triggered whenever the compaction counter of Level \(\ell'\) reaches zero. By the inductive hypothesis, the first such compaction occurs after 
\[
\binom{k+\ell'-1}{\ell'}
\]
buffer flushes. According to lines 12–13 of Algorithm 2, after a compaction, the compaction counters of all levels are reset to \(k-1\). Applying the same reasoning, the second compaction (which decrements the counter at Level \(\ell'+1\) from \(k-1\) to \(k-2\)) occurs after an additional 
\[
\binom{(k-1)+\ell'-1}{\ell'}
\]
buffer flushes. Iterating this argument, the \(i\)th compaction from Level \(\ell'\) to Level \(\ell'+1\) occurs after an additional 
\[
\binom{(k-i+1)+\ell'-1}{\ell'}
\]
buffer flushes. When the \(k\)-th compaction occurs, the compaction counter at Level \(\ell'+1\) reaches zero.

Thus, the total number of buffer flushes required is
\[
\sum_{i=1}^{k} \binom{(k-i+1)+\ell'-1}{\ell'} = \sum_{i=1}^{k} \binom{i+\ell'-1}{\ell'}.
\]
A well-known corollary of Pascal’s rule states that
\[
\sum_{i=1}^{k} \binom{i+\ell'-1}{\ell'} = \binom{k+\ell'}{\ell'+1}.
\]
This completes the inductive step, showing that the lemma holds for \(\ell = \ell'+1\).

By induction, the lemma holds for all \(\ell \ge 1\).
\end{proof}

\subsection{Complete Proof of Theorem~\ref{lemma: horizontal tiering optimality}}

To prove Theorem~\ref{lemma: horizontal tiering optimality}, we follow the sketch outlined in the main text. First, we show that the original problem \(\psi(n,\ell)\) may be decomposed into two subproblems—one corresponding to the writes before the first compaction into the highest level and the other to those afterward—by invoking Lemmas~\ref{lemma: level 1} and~\ref{lemma: uniform sub tiering}. Next, Lemma~\ref{lemma: best i tiering} establishes the precise timing of the first compaction to the highest level, denoted \(p_f^*\), in any optimal compaction sequence. By iteratively applying Lemmas~\ref{lemma: uniform sub tiering} and~\ref{lemma: best i tiering}, we recursively decompose the original problem into smaller subproblems. Finally, we demonstrate that the compaction sequence produced by Algorithm~\ref{alg: horizontal tiering} is identical to the optimal sequence for \(\psi(n,\ell)\).

\begin{lemma}
    In the optimal compaction sequence, all compactions start from Level 1.
    \label{lemma: level 1}
\end{lemma}

% \begin{proof}
%     Given a sequence of compactions $S = \{p_1, p_2, ...\}$. For any compaction $p_k=(I_k, l_1^k, l_2^k)$ in sequence $S$, if $l_1^k>1$, replacing $p_k$ with an alternate compaction $p_k^*=(I_k, 1, l_2^k)$ would always reduce the total read cost of this compaction sequence. This reduction occurs because the alternative compaction shortens the lifetime of all runs in the levels preceding level $l_1^k$ at the time $p_k$ occurs, without incurring additional costs through the creation of extra runs or extending the existence of any existing run.
% \end{proof}

\begin{proof}
Let \(S = \{p_1, p_2, \ldots\}\) be an optimal compaction sequence. Consider any compaction \(p_k = (I_k, l_1^k, l_2^k)\) in \(S\) with \(l_1^k > 1\). We construct an alternative compaction
\[
p_k^* = (I_k, 1, l_2^k),
\]
which is identical to \(p_k\) except that it starts from Level 1. Performing \(p_k^*\) instead of \(p_k\) shortens the lifetime of all runs in Levels \(1\) through \(l_1^k-1\) at the time of the compaction, thereby reducing the overall read cost incurred during subsequent lookups. Importantly, this modification does not introduce any additional overhead—for example, it does not create extra runs nor does it extend the duration that any run remains in the system. Consequently, replacing any compaction starting above Level 1 with one starting at Level 1 strictly lowers the total read cost, implying that every compaction in an optimal sequence must originate from Level 1.
\end{proof}

\begin{lemma}
Let $\tau(n, \ell)$ denote the total read cost of the optimal compaction sequence for problem $\psi(n,\ell)$. Let $r$ denote the number of lookups that occur between each buffer flush, we have~\footnote{In our analysis, for simplicity, we assume that there is no Bloom filter, meaning each lookup incurs one I/O for probing each run. However, this lemma can easily be extended to incorporate Bloom filters by multiplying the read cost in any given case by the corresponding false positive rate when using Bloom filters.}
\begin{equation}
\tau(n,\ell)=
\begin{cases}
0, & n=1,\\[1mm]
\displaystyle \binom{n}{2}\cdot r, & \ell=1,\; n>1,\\[1mm]
\displaystyle \min_{1\le i\le n-1} \Big\{\, \tau(i,\ell-1) + (n-i)\cdot r + \tau(n-i,\ell) \,\Big\}, & \text{otherwise}.
\end{cases}
\label{equation: uniform sub tiering}
\end{equation}
\label{lemma: uniform sub tiering}
\end{lemma}

\begin{proof}
We prove the lemma by considering the three cases.

\paragraph{Case 1: \(n=1\).}  
    The zero cost when $n=1$ owes to the fact that the first run in the LSM-tree is just created by buffer flush. Consequently, no previous lookups would have probed any run, resulting in no incurred read cost.

\paragraph{Case 2: \(\ell=1\) and \(n>1\).}  
When there is only one level, between the $i-1$-th and $i$-th buffer flush, there are $i-1$ runs in Level~1, and each lookup costs $(i-1)$ I/Os for probing all runs. Therefore, $(i-1)\cdot r $ I/Os would be incurred during this period.
    Summing these costs over all \(i\) from \(2\) to \(n\) yields
\[
\tau(n,1)=\sum_{i=2}^{n}(i-1)\cdot r = \binom{n}{2}\cdot r.
\]

\paragraph{Case 3: \(\ell>1\) and \(n>1\).} To prove the last case, we will illustrate that problem $\psi(n, \ell)$ can be decomposed into two subproblems. Specifically, suppose that in the optimal compaction sequence $S^* = \{p^*_1, p^*_2, \dots\}$, the first compaction to the largest level $p^*_f$ occurs after the $i$-th buffer flush. We can then partition the original problem into two subproblems. 
    Subproblem $\psi(i, \ell - 1)$: this subproblem addresses the first $i$ buffer flushes with the first $\ell - 1$ levels.
    Subproblem $\psi(n - i, \ell)$: this subproblem involoves the remaining $n - i$ buffer flushes and all $\ell$ levels.
    If the optimal compaction sequences for these subproblems are $S_1$ and $S_2$ respectively, then the optimal solution to the original problem satisfies $S^* = \{S_1, p^*_f, S_2\}$.
    We first explain why the portion of $S^*$ preceding $p^*_l$ must be $S_1$. Before the first compaction to the largest level in $S^*$, the data entering the LSM-tree with the first $i$ buffer flushes remains within the first $\ell - 1$ levels since no compaction to Level $\ell$ occurs before $p^*_f$. This scenario exactly matches problem $\psi(n - i, \ell)$
    because $S^*$ is the optimal compaction sequence that minimizes the total read cost, and the portion of $S^*$ preceding $p^*_f$ must correspond to the optimal solution of $\psi(i, \ell - 1)$. Otherwise, we could replace that portion with the optimal sequence $S_1$, resulting in a lower total read cost, contradiction ensues. %which would contradict the assumed optimality of $S^*$.
    Similarly, the portion of $S^*$ following $p^*_f$ must correspond to the optimal solution of $\psi(n - i, \ell)$, since the portion of data already moved into the largest level would not be merged with and affect any further data entering the LSM-tree. If it were not optimal, substituting it with $S_2$ would yield a compaction sequence with a lower read cost.

    Therefore, for the optimal compaction sequence, if the first compaction to the largest level $p^*_f$ occurs after the $i$-th buffer flush, the minimal total cost can be expressed as 
    \[\tau(i, \ell-1) + (n-i)\cdot r + \tau(n-i, \ell).\] Specifically, $\tau(i, \ell-1)$ represents the cost of the preceding subproblem, $\tau(n-i, \ell)$ represents the cost of the subsequent subproblem, and $(n-i)\cdot r$ corresponds to the read cost incurred by the run created by compaction $p^*_f$. Since the optimal compaction sequence yields the minimal read cost among all possible compaction sequences, $\tau(n, \ell)$ is the minimum over all possible values of $i$, given by $\min_{i \in [1, n-1]} \left\{ \tau(i, \ell-1) + (n-i)\cdot r + \tau(n-i, \ell) \right\}$.
\end{proof}

\begin{lemma}
For problem $\tau(n,\ell)$ with $n>1$ and $\ell>1$. Let $m$ be an integer that satisfies ${{m}\choose{\ell}} \leq n \leq {{m+1}\choose{\ell}}$. we have 
\begin{equation}
    \tau(n, \ell) = \tau(i^*, \ell-1)+(n-i^*)\cdot r+\tau(n-i^*, \ell)
    \label{equation: tiering decompose}
\end{equation} if $i$ satisfies
\begin{align}
    &{{m-1}\choose{\ell}}\leq i^*\leq{{m}\choose{\ell}}\\
    &{{m-1}\choose{\ell-1}}\leq n-i^*\leq{{m}\choose{\ell-1}}
\end{align}
\label{lemma: best i tiering}
\end{lemma}
We can prove the following lemma while simultaneously establishing Lemma~\ref{lemma: best i tiering}.

\begin{lemma}
    Let $m$ be an integer that satisfies ${{m}\choose{\ell}} \leq n \leq {{m+1}\choose{\ell}}$. Then $\tau(n, \ell)$ can be represented as
    \begin{equation}
    \tau(n, \ell) = \left[\ell \cdot {{m}\choose{\ell+1}} + (m-\ell+1)\cdot\left(n-{{m}\choose{\ell}}\right)\right]\cdot r
    \label{equation: uniform cost tiering}
    \end{equation}
    \label{lemma: horizontal tiering cost}
\end{lemma}

\begin{proof}
When $\ell=1$, $m$ equals to either $n$ or $n-1$. 
If $m=n$, the RHS of Equation~\ref{equation: uniform cost tiering} reduces to
\[
\left[1\cdot {{n}\choose{2}}+(n+1-1)\cdot(n-n)\right]\cdot r = {{n}\choose{2}}\cdot r.
\]
Otherwise when $m=n-1$, the RHS of Equation~\ref{equation: uniform cost tiering} also reduces to
\[
\left[1\cdot {{n-1}\choose{2}}+(n-1)\cdot(n-(n-1))\right]\cdot r = {{n}\choose{2}}\cdot r.
\]
According to Lemma~\ref{lemma: uniform sub tiering}, Equation~\ref{equation: uniform cost tiering} holds when $\ell = 1$.

When $n=1$, $m$ equals to either $\ell$ or $\ell-1$. 
If $m=\ell$, the RHS of Equation~\ref{equation: uniform cost tiering} reduces to
\[
\left[\ell\cdot 0 + (\ell+1)\cdot (1-1)\right]\cdot r = 0
\]
Otherwise when $m=\ell - 1$, the RHS of Equation~\ref{equation: uniform cost tiering} reduces to
\[
\left[\ell\cdot 0 + (\ell-\ell) \cdot (1-0)\right]\cdot r = 0
\]
According to Lemma~\ref{lemma: uniform sub tiering}, Equation~\ref{equation: uniform cost tiering} holds when $n = 1$. 

We will now use mathematical induction to prove that Equation~\ref{equation: uniform cost tiering} holds in the general case. 
When Equation~\ref{equation: uniform cost tiering} holds true for all $n<n'$ or $\ell<\ell'$,
our objective is to prove it also holds for $\tau(n',\ell')$.
Let $i^*$ be an integer that satisfies
\begin{align*}
    {{m-1}\choose{\ell}}\leq i^*\leq{{m}\choose{\ell}}\\
    {{m-1}\choose{\ell-1}}\leq n-i^*\leq{{m}\choose{\ell-1}}
\end{align*}
There is always at least one valid value for $i^*$, because
\[
{{m-1}\choose{\ell}}+{{m-1}\choose{\ell-1}}={{m}\choose{\ell}},\ and\ {{m}\choose{\ell}}+{{m}\choose{\ell-1}}={{m+1}\choose{\ell}}
\]
According to Lemma~\ref{lemma: uniform sub tiering}, we have
\begin{align*}
&\tau(n,\ell)\leq \tau_(i^*, \ell - 1)+(n-i)\cdot \tau+\tau_(n-i^*, \ell)\\
=& \left[(\ell-1) \cdot {{m-1}\choose{\ell}} + (m - 1 - (\ell - 1)+1)\cdot\left(i-{{m-1}\choose{\ell-1}}\right) \right]\cdot r\\
&+(n-i^*)\cdot r \left[(\ell \cdot {{m-1}\choose{\ell + 1}} + (m-\ell +2)\cdot\left(n-i^*-{{m-1}\choose{\ell}}\right) \right]\cdot r\\
=&\left[\ell\cdot\left[{{m-1}\choose{\ell}}+{{m-1}\choose{\ell+1}}\right]+(m-\ell)\cdot\left(n-{{m-1}\choose{\ell-1}}-{{m-1}\choose{\ell}}\right)\right]\cdot r\\
%&-{{m}\choose{\ell}}-\left(i-{{m-1}\choose{\ell}}\right) + i \\
&+\left[n-{{m-1}\choose{\ell-1}}- {{m-1}\choose{\ell}} \right]\cdot r\\
=&\left[\ell \cdot {{m}\choose{\ell+1}} + (m-\ell+1)\cdot\left(n-{{m}\choose{\ell}}\right)\right]\cdot  r
\end{align*}

The above inequality gives that the RHS of Equation~\ref{equation: uniform cost tiering} is an upper bound of $\tau(n, \ell)$. It remains to prove that the RHS of  Equation~\ref{equation: uniform cost tiering} is also a lower bound of $\tau(n, \ell)$.
Specifically, we will show that, for any $j\in [1,n-1]$, 
\[
\tau(i^*, \ell-1)+(n-i)^*\cdot r+\tau(n-i^*, \ell) \leq \tau(j, \ell-1)+(n-j)\cdot r+\tau(n-j, \ell),
\]assuming $j$ satisfies that 
\begin{align*}
    {{a-1}\choose{\ell-1}}&\leq j\leq{{a}\choose{\ell-1}}\\
    {{b-1}\choose{\ell}}&\leq n-j\leq{{b}\choose{\ell}}
\end{align*}
Because ${{a-1}\choose{\ell}}+{{b-1}\choose{\ell-1}}=n$, when $a \leq m$, we have $b \geq m$.
This gives us 
{
\begin{align*}
&\tau(j, \ell-1)+(n-j)\cdot r+\tau(n-j, \ell) - \tau(i^*, \ell-1)-(n-i^*)\cdot r-\tau(n-i^*, \ell) \\
= &\left[\ell\cdot  {{b}\choose{\ell+1}} + (b-\ell+1)\left(n-j-{{b-1}\choose{\ell}}\right)\right]\cdot r\\
-& \left[\ell\cdot  {{m}\choose{\ell+1}} - (m-\ell+1)\left(n-i^*-{{m-1}\choose{\ell}}\right) -  (i^*-j) \right]\cdot r\\
-& \left[(\ell-1)  {{m}\choose{\ell}} + (m-\ell)\left(i^*-{{m-1}\choose{\ell-1}}\right)
- (\ell-1)  {{a}\choose{\ell}} - (a-\ell)\left(j-{{a-1}\choose{\ell-1}}\right)\right]\cdot r\\
\geq &[(m+1) (i^*-j) - (m+1) (i^* - j)]\cdot r = 0.
\end{align*}}
Similarly, when $a>m$, we can derive the same result with analogous reasoning.
Therefore, the RHS of Equation~\ref{equation: uniform cost tiering} is also a lower bound of $\tau(i,\ell)$.
Hence, Equation~\ref{equation: uniform cost tiering} always holds.
This also demonstrates that no other value of $i$ in Equation~\ref{equation: uniform sub tiering} can result in a smaller value than when $i = i^*$, thereby proving Lemma~\ref{lemma: best i tiering}.
\end{proof}

\begin{proof}[Proof of Theorem~\ref{lemma: horizontal tiering optimality}]
When \( N = \binom{k + \ell - 1}{\ell} \cdot B \), the corresponding problem is \( \psi\left( \binom{k + \ell - 1}{\ell}, \ell \right) \). We define \( S_h(N, \ell) \) as the horizontal-tiering compaction sequence for this problem, which is the sequence adopted by Algorithm~\ref{alg: horizontal tiering}.
Our proof proceeds as follows: (1). Lemma~\ref{lemma: best i tiering} indicate that the timing \( i_f \) of the first compaction \( p_f \) to the largest level in \( S_h(N, \ell) \) is optimal for the problem \( \psi\left( \binom{k + \ell - 1}{\ell}, \ell \right) \).
(2). Let \( S_1 \) be the portion of \( S_h(N, \ell) \) before \( p_f \), and \( S_2 \) be the portion after \( p_f \). We show that:
\[
S_1 = S_h\left( \binom{k + \ell - 2}{\ell - 1} \cdot B, \ell - 1 \right) = S_h\left( i_f \cdot B, \ell - 1 \right) ,
\]
and
\[
S_2 =S_h\left( \binom{k + \ell - 2}{\ell } \cdot B, \ell \right) = S_h\left( N - i_f \cdot B, \ell \right).
\]
This indicates that \( S_1 \) and \( S_2 \) correspond to the horizontal-tiering compaction sequences for their respective data sizes and levels.
(3).  For the two subproblems resulting from partitioning \( \psi\left( \binom{k + \ell - 1}{\ell}, \ell \right) \) at \( i_f \)—specifically, \( \psi\left( i_f, \ell - 1 \right) \) and \( \psi\left( n - i_f, \ell \right) \)—the sequences \( S_1 \) and \( S_2 \) are their horizontal-tiering compaction sequences.

By combining these points, we can recursively demonstrate that all compactions in the sequence produced by Algorithm~\ref{alg: horizontal tiering} correspond one-to-one with those in the optimal compaction sequence of the problem \( \psi\left( n, \ell \right) \). Therefore, Algorithm~\ref{alg: horizontal tiering} yields the optimal compaction sequence for the given problem.

\end{proof}

\subsection{Proof of Lemma~\ref{lemma: rt}}
\begin{proof}
Theorem~\ref{lemma: horizontal tiering optimality} illustrated that with the horizontal-tiering algorithm, the horizontal part in {\ourmethod} adopts the optimal compaction sequence. Therefore Equation~\ref{equation: rt} always holds according to Lemma~\ref{lemma: horizontal tiering cost}.
\end{proof}

\subsection{Proof Sketch of Lemma~\ref{lemma: wl}}
Under the leveling merge policy, as data entries flux in, each buffer flush to Level 1 incurs an I/O cost of \( \frac{D_1(i)}{P} \), where \( D_1(i) \) represents the size of Level 1 after the buffer flush, and \( P \) is the disk page size. Between consecutive buffer flushes, compactions can be performed between adjacent levels to control the size of Level 1 and other levels, reducing future write costs. A compaction that starts from Level \( l_1 \) and ends at Level \( l_2 \) incurs a cost of \( \sum_{j=l_1}^{l_2} \frac{D_j(i)}{P} \), merging all data from Levels \( l_1 \), \( l_1+1 \), ..., \( l_2-1 \) into Level \( l_2 \). Levels \( l_1 \) through \( l_2-1 \) into Level \( l_2 \).
Different growth schemes result in varying compaction timings, which ultimately lead to different costs. For example, if compactions are too infrequent or not performed at all, data accumulates in Level 1, making each buffer flush increasingly costly. 
The timing of each compaction can be represented by a triple \( S = (I, l_1, l_2) \), where \( I \) denotes that the compaction occurs after the \( I \)-th buffer flush, and \( l_1 \) and \( l_2 \) indicate the starting and ending levels, respectively. This scenario can then be formulated into the following problem.
\footnote{Note that under the leveling merge policy, when the compaction counters of several consecutive levels are equal, a compaction to the initial level will sequentially trigger compactions at each subsequent level. In such cases, we can merge these consecutive compactions into a single, unified compaction, which will slightly reduce the overall write amplification. Specifically, this approach reduces the write amplification by one at each level because each entry will be involved in exactly one consecutive compaction as it passes through each level—that is, the compaction that triggers the compaction from that level to the next originates from the compaction at the preceding level. In our analyses, we assume that this optimization is employed.}
% \footnote{Under the leveling merge policy, compaction from a level may sequentially trigger compactions at subsequent levels. We can merge these into a single unified compaction, slightly reducing write amplification by one at each level, since each entry will be involved in exactly one consecutive compaction as it passes through each level—that is, the compaction that triggers the compaction from that level to the next originates from the compaction at the preceding level. In our analyses, we assume this optimization is employed.} 

\noindent\begin{problem}
    Consider a workload that would incur a total of \( n \) sequential buffer flushes in the LSM-tree, which has a maximum of \( \ell \) levels. We define the problem as $\psi(n, \ell)$.
    % Specifically, the cost associated with the \( i \)-th buffer flush is \( \frac{D_1(i)}{P} \), where \( D_1(i) \) denotes the size of Level 1 after the \( i \)-th buffer flush, and \( P \) is the disk page size.
    % Between any two consecutive buffer flushes \( i \) and \( i+1 \), we can perform compactions between a series of adjacent levels. A compaction that starts from Level \( l_1 \) and ends at Level \( l_2 \) incurs a cost of \( \sum_{j=l_1}^{l_2} \frac{D_j(i)}{P} \), merging all data from Levels \( l_1 \), \( l_1+1 \), ..., \( l_2-1 \) into Level \( l_2 \).
    Our objective is to minimize the total write cost of processing this workload by determining an optimal compaction sequence \( \{S_1, S_2, \dots\} \). 
    % Each compaction \( S \) is represented by a triple \( (I, l_1, l_2) \), indicating that this compaction takes place after the \( I \)-th buffer flush and starts from Level \( l_1 \) and ends at Level \( l_2 \).
\label{problem: uniform leveling}
\end{problem}

After formulating the cost of different compaction sequences of the LSM-tree under leveling merge policy into Problem 2
We can prove Lemma~\ref{lemma: wl} following similar procedures in the proof of Theorem~\ref{lemma: horizontal tiering optimality} and Lemma~\ref{lemma: rt}.

\subsection{Derivation of Equation~\ref{equation: delta} in Section~\ref{Sec: skewed}}

Assume that we perform a total of \(k\) compactions in the first level, with the current threshold set at \(\delta-1\). In this configuration, the first compaction in the sequence occurs after \(\delta\) buffer flushes. Increasing the threshold from \(\delta-1\) to \(\delta\) effectively removes the compaction triggered after the first \(\delta\) flushes and increments the required number of flushes for the final \(\delta\) compactions by one. Consequently, this threshold increment reduces the write cost by eliminating one compaction; however, it also introduces additional overhead due to the increased cost of the final \(\delta\) compactions. This benefit can be represented as

\[
B = \alpha + \delta\cdot(1-\alpha) + \delta\cdot(1-\alpha)\cdot(k-1) + \sum_{i=1}^{\delta}[(\alpha + i\cdot(1-\alpha)],
\]

while the corresponding cost is represented as

\[
C = \sum_{i=1}^{\delta}\left[\alpha + (\delta+k-i+1)\cdot(1-\alpha)\right] + \sum_{i=1}^{\delta} i\cdot(1-\alpha).
\]

Subtracting the cost from the benefit, we obtain the following net change in write cost:  
\begin{align*}
\Delta {\text{cost}} & = B - C \\ 
%\left[\alpha + \delta(1-\alpha) + \delta\cdot(1-\alpha)\cdot(k-1) + \sum_{i=1}^{\delta}\left(\alpha + i\cdot(1-\alpha)\right)\right] - \left[\sum_{i=1}^{\delta}\left[\alpha + (\delta+k-i+1)(1-\alpha)\right] + \sum_{i=1}^{\delta} i\cdot(1-\alpha)\right]\\
& = \frac{\delta\cdot(\delta-1)}{2}\cdot(1-\alpha)-\delta^2\cdot(1-\alpha)+\alpha\\
& = \alpha - \frac{\delta\cdot(\delta+1)}{2}\cdot(1-\alpha).
\end{align*}

\end{document}